%% file: arxiv-v2.tex
\documentclass[pra,superscriptaddress,notitlepage,twocolumn,nofootinbib]{revtex4-1}
\usepackage{times}
\input{pretex}

\definecolor{beamer}{rgb}{0.2,0.2,0.7}
\definecolor{colorone}{rgb}{1,0.36,0.03}
\definecolor{colortwo}{rgb}{0.4,0.77,0.17}
\definecolor{colorthree}{rgb}{0.01,0.51,0.93}
\definecolor{colorfour}{rgb}{0.47,0.26,0.58}
\definecolor{colorfive}{rgb}{0.12,0.55,0.16}
\usepackage{tcolorbox}
\usepackage{relsize}

\usepackage{tikz-cd}
\tikzcdset{every label/.append style = {font = \small}}
\usepackage{tikz}
\usetikzlibrary{positioning}
\usetikzlibrary{shapes.geometric}
\usetikzlibrary{calc}
\definecolor{tensorblue}{rgb}{0.8,0.9,1}
\tikzset{ten/.style={fill=tensorblue}}
\newcommand{\diagram}[1]{ \begin{array}{cc}\begin{tikzpicture}[scale=.5,every node/.style={sloped,allow upside down},baseline={([yshift=+0ex]current bounding box.center)}] #1 \end{tikzpicture} \end{array} }

\newcommand{\BCNOT}{\widetilde{\CNOT}}
\newcommand{\ansatz}{SEA}
\newcommand{\sfid}{F(\ket{\phi}, \ket{\psi})}

\newcommand{\bin}[2]{\Bar{#1}_{#2}}

\allowdisplaybreaks

\begin{document}
	\title{Mitigating barren plateaus of variational quantum eigensolvers}
	\author{Xia Liu}
	\affiliation{Institute for Quantum Computing, Baidu Research, Beijing 100193, China}
	\author{Geng Liu}
	\affiliation{Institute for Quantum Computing, Baidu Research, Beijing 100193, China}
	\author{Jiaxin Huang}
	\affiliation{Institute for Quantum Computing, Baidu Research, Beijing 100193, China}
	\author{Hao-Kai Zhang}
	\affiliation{Institute for Quantum Computing, Baidu Research, Beijing 100193, China}
	\affiliation{Institute for Advanced Study, Tsinghua University, Beijing 100084, China}
	\author{Xin Wang}
	\email{wangxin73@baidu.com}
	\affiliation{Institute for Quantum Computing, Baidu Research, Beijing 100193, China}
	
	\begin{abstract}
		Variational quantum algorithms (VQAs) are expected to establish valuable applications on near-term quantum computers. However, recent works have pointed out that the performance of VQAs greatly relies on the expressibility of the ansatzes and is seriously limited by optimization issues such as barren plateaus (i.e., vanishing gradients). This work proposes the state efficient ansatz (SEA) for accurate ground state preparation with improved trainability. We show that the SEA can generate an arbitrary pure state with much fewer parameters than a universal ansatz, making it efficient for tasks like ground state estimation. Then, we prove that barren plateaus can be efficiently mitigated by the SEA and the trainability can be further improved most quadratically by flexibly adjusting the entangling capability of the SEA. Finally, we investigate a plethora of examples in ground state estimation where we obtain significant improvements in the magnitude of cost gradient and the convergence speed. 
	\end{abstract}
	\date{\today}
	\maketitle
	\textbf{\textit{Introduction.}}---
	Quantum computers are expected to achieve quantum advantages~\cite{preskill2012quantum} in solving valuable problems
	\cite{lloyd1996universal,harrow2009quantum,Childs2010,Montanaro2016}. Before arriving at universal quantum computing, a key direction is to explore the power of noisy intermediate-scale quantum (NISQ)~\cite{preskill2018quantum} devices in important fields such as quantum chemistry~\cite{cao2019quantum,mcardle2020quantum} and quantum machine learning~\cite{Schuld2021,Biamonte2017a,Huang2021,Jerbi2021,Mitarai2018}.
	Recent results~\cite{arute2019quantum,Wu2021a,Zhong2020,zhong2021phase, Abbas2021,madsen2022quantum} have shown that quantum advantages in specific tasks can be achieved using such devices.
	
	One common paradigm for designing quantum solutions using NISQ devices is variational quantum algorithms (VQAs) ~\cite{cerezo2021variational,Bharti2021,Endo2020}. VQAs are promising to deliver applications in many important topics, including ground state preparation~\cite{peruzzo2014variational, cao2019quantum, mcardle2020quantum}, quantum data compression~\cite{Romero2017,Wang2020d}, machine learning~\cite{Cong2018,Li2021-classifier,Schuld2021,li2021vsql,Caro2021,Farhi2018}, and combinatorial optimization~\cite{farhi2014quantum,Zhou2018b}. Combining the advantages of classical computers and quantum devices, VQAs adopt parameterized quantum circuits (PQCs, also known as ansatzes)~\cite{benedetti2019parameterized} and utilize classical computers to optimize the parameters to minimize the cost functions that are designed for solving target problems.
	
	Among the numerous VQAs, the variational quantum eigensolver (VQE) for ground state estimation is a central one of both practical and theoretical interests. For VQE, one common approach is to utilize fixed structure ansatzes, such as the hardware-efficient ansatz~\cite{kandala2017hardware} and the unitary coupled cluster ansatz~\cite{hoffmann1988unitary, mcardle2020quantum, lee2018generalized}, which require a large depth to achieve high accuracy when the problem scale is large. On the other hand, there are adaptive structure ansatzes~\cite{grimsley2019adaptive, tang2021qubit,ryabinkin2019iterative,Bharti2021a,Grimsley2022}, which usually suffer from relatively high costs of both quantum and classical resources.
	
	\begin{figure}[t]
		\centering
		\includegraphics[width=0.45\textwidth]{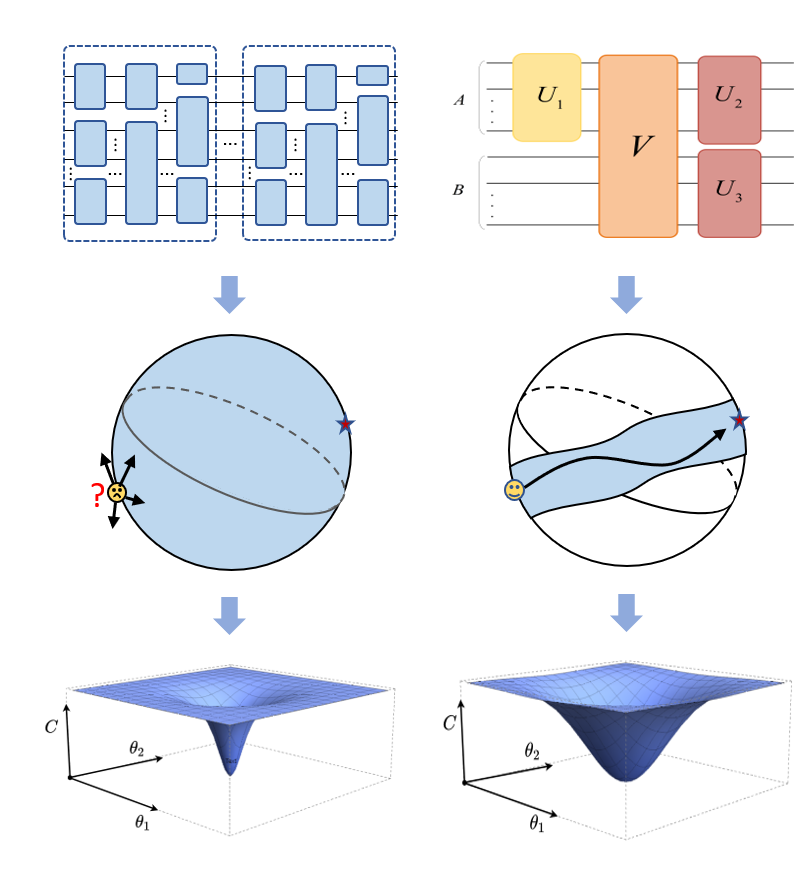}
		\caption{\textbf{Illustration of our main result.} The left column demonstrates the difficulty in the optimization process of general ansatzes, which have high expressibility. The optimization process of State Efficient Ansatz ({\ansatz}) is depicted in the right column, with its conceptual framework represented in the upper right, where $A$ and $B$ are two subsystems. Schmidt coefficient layer $U_1$, entangling layer $V$, and local basis changing layer $U_2, U_3$, are the three parts of {\ansatz}. The key feature of SEA is reducing the optional range of paths, leading to a more favorable landscape, as shown at the bottom.}
	\label{fig:TA-structure}
\end{figure}

Despite the success of VQE, expressibility~\cite{haug2021capacity,Du2022,holmes2022connecting} and trainability~\cite{mcclean2018barren,zhang2021toward,Zhang2022} are still two critical challenges of designing ansatzes for it. To improve the accuracy of the algorithms, deep ansatzes are usually preferred for their strong expressibility.
This could be seen in the same spirit of deep neural networks in machine learning. However, recent works~\cite{holmes2022connecting,sim2019expressibility,nakaji2021expressibility} imply that strong expressibility will lead to poor trainability due to the barren plateau phenomenon~\cite{mcclean2018barren, Sharma2020c, Marrero2020}, i.e., exponentially vanishing gradients in training with respect to the number of qubits.
This technical bottleneck seriously restricts the scalability of VQE. 
Despite a plethora of recent attempts to address the barren plateaus through adjusting initialization~\cite{Grant2019,Verdon2019}, cost functions~\cite{cerezo2021cost,kieferova2021quantum}, and architectures~\cite{Pesah2021}, there is still a strong need for extending the scalability of VQE by improving the trainability while preserving the effectiveness.


To overcome these challenges, we propose the State Efficient Ansatz ({\ansatz}), as shown in Fig.~\ref{fig:TA-structure}.
By removing the redundancy between the sets of universal unitary and universal pure quantum state, {\ansatz} can represent an arbitrary pure quantum state with fewer parameters, resulting in its efficiency in all the tasks that are essentially learning a quantum state, which we call the state-oriented tasks (such as VQE). 
Moreover, by adjusting the number of CNOT gates in the entangling layer (see Fig.~\ref{fig:TA-CNOT}), SEA has the ability to learn a target ground state with low bipartite entanglement, where we proved that barren plateaus can be effectively mitigated by the SEA. 
To be specific, the SEA can effectively enhance the gradient magnitude up to a square root of the scaling of the global $2$-design case. 
We further investigate a plethora of examples in ground state estimation through numerical experiments and establish evident improvements 
in both the overall behavior and the magnitude of the cost gradient .


\textbf{\textit{State Efficient Ansatz.}}---\label{sec: SEA}
The intuition of {\ansatz} is to prune the redundant expressibility of an ansatz in the state-oriented tasks by utilizing the Schmidt representation~\cite{Nielsen2010} of pure states, which also inspires the tree tensor network~\cite{shi2006classical, nagaj2008quantum, tagliacozzo2009simulation, nakatani2013efficient}. Specifically, the input system has two subsystems $A$ and $B$, and the ansatz comprises three parts. The first part is the \textit{Schmidt coefficient layer} $U_1$ acting on subsystem $A$. The second part is the \textit{entangling layer} $V$ acting on the whole system $AB$ to create entanglement between two subsystems. The last part is the \textit{local basis changing layer (LBC layer)}, which consists of two local circuits, $U_2$ and $U_3$, applied to the two subsystems respectively.
The whole structure of {\ansatz} is shown in the upper right of Fig.~\ref{fig:TA-structure} and can be written as 
\begin{equation}
S(\bm{\theta})\equiv(U_2(\bm{\theta_2})\otimes U_3(\bm{\theta_3}))V(U_1(\bm{\theta_1})\otimes I),
\label{eq:TA-general}
\end{equation}
where $S$ is the unitary representation of {\ansatz}, $\bm{\theta}=\{\bm{\theta_1},\bm{\theta_2},\bm{\theta_3}\}$, and each $\bm{\theta_i}(i=1,2,3)$ is a parameter vector. 

Now we explain how SEA reduces redundant expressibility from the perspective of degrees of freedom (DOF). For state-oriented tasks, the DOF of the $2N$-qubit pure state set is $2^{2N+1}-2$ and the DOF of the unitary group $\mathcal{U}(2^{2N})$ is $4^{2N}$~\cite{osti_4271130}. Therefore, using a universal $2N$-qubit PQC that can represent a universal unitary set to generate a universal pure states set is overkill. In contrast, SEA has a quadratic advantage in the parametric DOF compared to universal PQCs. Note that a universal $N$-qubit PQC $U_i$ has $O(4^N)$ DOF~\cite{PhysRevA.103.L030401}. Therefore, a $2N$-qubit {\ansatz} with universal $U_i(i=1,2,3)$ also has $O(4^N)$ parametric DOF. In comparison, a $2N$-qubit universal PQC has $O(4^{2N})$ parametric DOF.
Therefore, {\ansatz} requires quadratically fewer parameters than a general PQC of the same dimension while retaining the ability to generate an arbitrary pure state. It is also worth pointing out that SEA can reach the optimal DOF to generate pure states set, which is the $2\times 4^N -2$, by removing the redundant freedom in $U_1,U_2,U_3$. In this sense, we call it \textit{state efficient ansatz}. 

After clarifying that SEA is efficient in generating an arbitrary pure state, we here elaborate the effectiveness of {\ansatz}. For simplicity, we assume both subsystems $A$ and $B$ have $N$ qubits and let $\ket{0}^{\otimes 2N}$ be an initial state sent into {\ansatz}. An example of SEA is shown in Fig.~\ref{fig:TA-CNOT}. To be specific, $U_1$ is the Schmidt coefficient layer with $U_1\ket{0}^{\otimes N}=\sum_{k=0}^{2^N-1}\lambda_k\ket{k}_A$, where $\{\ket{k}\}_{k=0}^{2^N-1}$ is the computational basis. At the entangling layer, we set a composition of $N \CNOT$s controlled and targeted on the qubit-pairs $\{(i, N+i)\}_{i=0}^{N-1}$. At the last layer, we apply $U_2$ and $U_3$ on subsystems $A$ and $B$, respectively, such that $U_2\ket{k}_A=\ket{v_k}_A$ and $U_3\ket{k}_B=\ket{v_k}_B$. Then after SEA, $\ket{0}^{\otimes 2N}$ will evolve into $\sum_{k=0}^{2^N-1}\lambda_k\ket{v_k}_A\ket{v_k}_B$. Note that this exactly forms a Schmidt decomposition with tunable Schmidt coefficients and bases. Since any pure state has a form of Schmidt decomposition, {\ansatz} can evolve an initial state $\ket{0}^{\otimes 2N}$ into an arbitrary $2N$-qubit pure state if $U_1$ can generate an arbitrary $N$-qubit pure state (that is, $U_1$ possessing \textit{universal wavefunction expressibility}) and $U_2, U_3$ are universal. The whole process of SEA acting on the initial state $\ket{0}^{\otimes 2N}$ is shown in Fig.~\ref{fig:TA-CNOT}.

\begin{figure}[t]
\centering
\includegraphics[width=0.35\textwidth]{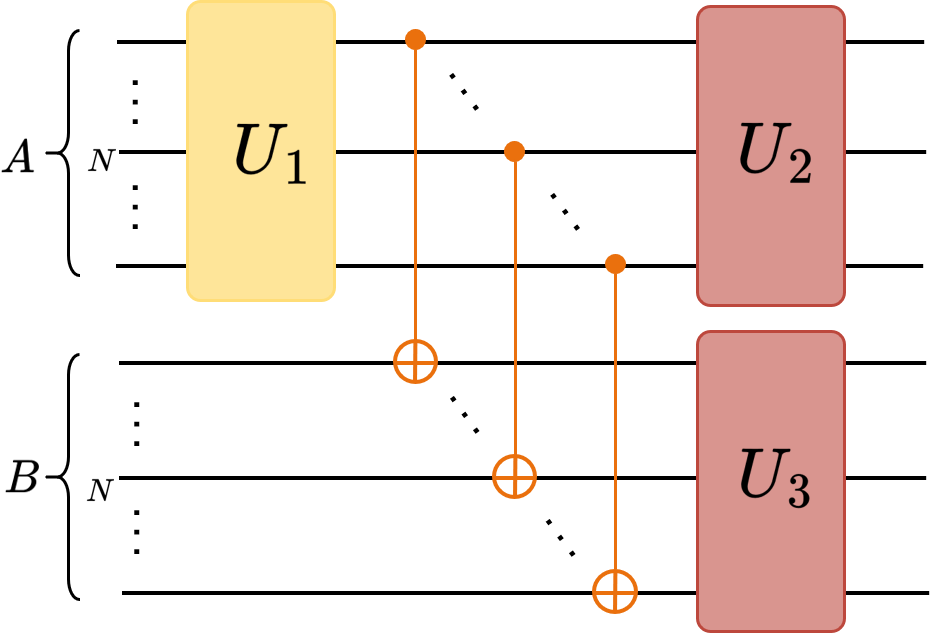}
\caption{\textbf{An established case of {\ansatz}.} Schmidt coefficient layer $U_1$ has universal wavefunction expressibility. The entangling layer consists of N CNOTs, and the LBC layer comprises two local universal PQCs, $U_2$ and $U_3$. The evolution of $\ket{0}^{\otimes 2N}$ under this SEA is as follows. After Schmidt coefficient layer, we obtain $\sum_{k=0}^{2^N-1}\lambda_k\ket{k}_A\ket{0}^{\otimes N}$. Then we get $\sum_{k=0}^{2^N-1}\lambda_k\ket{k}_A\ket{k}_B$ after implementing $N$ CNOTs. Finally, we obtain the output state $\sum_{k=0}^{2^N-1}\lambda_k\ket{v_k}_A\ket{v_k}_B$ after the LBC layer.}
\label{fig:TA-CNOT}
\end{figure}

The above analysis implies that {\ansatz} with a $U_1$ of the universal wavefunction expressibility and $U_2, U_3$ being universal can accurately solve state-oriented tasks. Note that an $N$-qubit {\ansatz} has the universal wavefunction expressibility under the same conditions, so we can use it as the Schmidt coefficient layer of a $2N$-qubit {\ansatz} to generate an arbitrary $2N$-qubit pure state.

According to SEA's ability to evolve $\ket{0}^{\otimes 2N}$ into an arbitrary $2N$-qubit pure state, we can conclude that SEA is effective for the VQE, which aims to variationally find the ground state of a given Hamiltonian.
To be specific, for any $2N$-qubit Hamiltonian $H$, the goal of the VQE is to find the ground state energy $E_0=\min_{\ket{\psi}} \bra{\psi}H\ket{\psi}$, where the minimization is over the $2N$-qubit pure state set. Since a $2N$-qubit {\ansatz} could generate an arbitrary pure state, it can obtain $E_0$ of a Hamiltonian after the optimization. This exactly illustrates the effectiveness of {\ansatz} when used in VQE. More details can be found in Sec.~II
of the Supplemental Material.

When dealing with a specific task, we can further optimize SEA if we know certain information about the entanglement of the target quantum state. The SEA mentioned before can generate a pure state with full Schmidt rank. However, if the target quantum state is weakly entangled, we can change the structure of the entangling layer to reduce the cost without sacrificing performance. The effectiveness of SEA, in this case, is described in the following proposition.


\begin{proposition}
\label{pro: truncation-general}
If $U_1$ can generate any $N$-qubit pure state that is a superposition of at most $K$ computational basis states, then for any $\ket{\phi}$, there exists an {\ansatz} output state $\ket{\psi}$
with $F(\ket{\phi},\ket{\psi})\geq \min\left\{\frac{K}{r}, 1\right\}$,
where $F(\ket{\phi},\ket{\psi})$ is the fidelity between $\ket{\phi}$ and $\ket{\psi}$, and $r$ is the Schmidt rank of $\ket\phi$.
\end{proposition}

We can tune the value $K$ in Proposition~\ref{pro: truncation-general}, which corresponds to the entangling capability, by adjusting the number of CNOT gates in the entangling layer. This proposition guarantees that SEA with fewer CNOT gates still has a strong expressibility for learning weakly entangled quantum states.
More experiments are shown in \textit{Numerical Simulation of Experiments}. 

So far, our analysis of SEA's effectiveness assumes perfect training. However, in practice, trainability is also an important factor to be considered. Many common ansatzes suffer from the notorious barren plateaus phenomenon~\cite{mcclean2018barren}. It was shown that the cost gradient vanishes exponentially in the number of qubits for a randomly initialized PQC with sufficient depth, if the unitary ensemble generated by the PQC is sufficiently random to accord with the Haar distribution over the unitary group up to the second moment, i.e., forming a unitary $2$-design. The following proposition implies that SEA does not form a unitary 2-design even if it has enough expressibility for state-oriented tasks.

\begin{proposition}
\label{pro:non-2-design}
{\ansatz} with $U_i(i=1,2,3)$ being local 2-design and $\BCNOT$ as entangling layer does not form a 2-design on the global system.
\end{proposition}

As we have known that a large class of random PQCs will end up forming a 2-design when the depth is large \cite{mcclean2018barren} and this is the usual case in VQAs,
Proposition~\ref{pro:non-2-design} implies that even though the local structures of SEA have strong expressibility by forming 2-designs, SEA will still not form a global 2-design. Thus SEA could avoid the known definite zone of barren plateaus without losing effectiveness on VQE. Moreover, SEA will maintain this property even if we change the number of CNOTs in the entangling layer. In fact, with less expressibility, SEA will be farther from being a 2-design. The proof can be found in Sec.~IV
of the Supplemental Material.

However, not being a unitary 2-design does not guarantee good trainability. In the following, we conduct a more rigorous analysis of SEA's trainability.

\textbf{\textit{Mitigating Barren Plateaus.}}---\label{sec:TA2-design}
In this section we estimate the trainability and expressibility of the SEA by computing the variance of the cost gradient~\cite{mcclean2018barren,cerezo2021cost} and the frame potential~\cite{sim2019expressibility}, respectively, given that the sub-blocks in the SEA are sampled from local unitary $2$-design ensembles. We show that by tuning the number of $\CNOT$ gates in the entangling layer, the SEA could continuously change the entanglement of the prepared state, and hence flexibly balance the trade-off between trainability and expressibility~\cite{holmes2022connecting}. Especially, if the target state is weakly entangled, the SEA could gain a square root advantage over the common circuit ansatzes such as the hardware-efficient ansatz, etc.

Trainability is one of the most concerning problems in VQAs. 
If the adopted ansatz forms a unitary $2$-design, the optimization process will be crippled due to the exponentially vanishing gradients, resulting in a severe issue on the trainability of VQAs. Specifically, we consider the VQE cost function $C(\bm{\theta})$ with an initial state $\ket{0}^{\otimes 2N}$, an objective operator $H$ and a circuit ansatz $\mathbf{U}(\bm{\theta})$ on $2N$ qubits and the gradient component $\partial_\mu C=\frac{\partial C}{\partial \theta_{\mu}}$ with respect to the variational parameter $\theta_\mu$. Here we assume that all sub-blocks with variational parameters in the SEA, i.e. $U_1,U_2$, and $U_3$, are sampled from local unitary $2$-designs so that we can integrate them using the Weingarten calculus~\cite{Collins2006,Puchaa2017}. Note that if $\theta_\mu$ locates at $U_i$, we assume that the two parts in $U_i$ are split by the gate corresponding to $\theta_\mu$ (like in Fig.~\ref{fig:RU_QNN}) are both local $2$-designs. The trainability of the SEA under the above assumption is described by the following proposition, the proof of which is presented in Sec.~V
of the Supplemental Material.
\begin{proposition}\label{prop:gradient-variance}
For an SEA defined on $2N$ qubits with all sub-blocks being local $2$-designs, the variance of the cost gradient scales with the number of qubits as
\begin{equation}
	\var_{\rm SEA}[\partial_\mu C]\in\mathcal{O}(2^{-(N+M)}),
\end{equation}
where $M\in\{0,1,...,N\}$ denotes the number of $\CNOT$ gates used in the entangling layer and the variance is taken over all SEA sub-blocks independently.
\end{proposition}

This result notably implies that, by tuning the number of CNOT gates in the entangling layer, the SEA can effectively enhance the gradient magnitude up to a square root of the scaling of the global $2$-design case $\var_{\rm Haar}[\partial_\mu C]\in\mathcal{O}(2^{-2N})$,
which is verified by the numerical results in Fig.~\ref{Fig: BP-sub.1} and Fig.~\ref{fig:BP_specific}. This is achieved by sacrificing the ability of expressing the highly entangled states since the number of CNOT gates determines the upper bound to the Schmidt rank of the prepared state. In other words, if the target state is actually low-entangled,
e.g., the Schmidt rank does not scale or scales slowly with the number of qubits, the SEA could effectively mitigate barren plateaus and at the same time keep the capability of expressing the target state. For the VQE focusing on the ground state preparation task, this low-entanglement assumption is reasonable and supported by the entanglement area law for the ground states of gapped systems and the logarithmic growth for the ground states of gapless systems, at least in the $1$-dimensional case~\cite{Eisert2008}.

The trade-off between the trainability and the expressibility~\cite{holmes2022connecting} motivates us to further quantitatively analyze the expressibility of SEA using the $t$-degree frame potential~\cite{sim2019expressibility}, which is defined by
\begin{equation}
\mathcal{F}^{(t)} = \E \left[ \left| \braandket{0}{\mathbf{U}^\dagger \mathbf{V}}{0} \right|^{2t} \right],
\end{equation}
where $\ket{0}$ denotes the zero state of the whole system and the expectation with respect to $\mathbf{U},\mathbf{V}$ is taken over two copies of the ansatz ensemble independently.
The $t$-degree frame potential measures the expressibility of the ansatz ensemble up to the $t$-th moment. It is known that the Haar random ensemble achieves the minimum frame potential $\mathcal{F}^{(t)}_{\rm Haar}$ for each $t$, and this value can be achieved if and only if the ansatz ensemble is a $t$-design ensemble~\cite{renes2004symmetric}. In addition, a smaller frame potential implies a stronger expressibility. Under the assumption of local $2$-designs, we have the following proposition for the expressibility of the SEA. The proof can be found in Sec.~VI
of the Supplemental Material.


\begin{proposition}\label{prop:frame-potential}
For an SEA defined on $2N$ qubits with all sub-blocks being local $2$-designs, the first and second frame potential satisfies
\begin{align}
	&\mathcal{F}^{(1)}_{\rm SEA} = 2^{-2N},\\
	&\mathcal{F}^{(2)}_{\rm SEA} \in \mathcal{O}(2^{-4N}).
\end{align}
\end{proposition}

The first frame potential of the SEA is identical to that of a global state $1$-design $\mathcal{F}^{(1)}_{\rm Haar}=2^{-2N}$ and the second frame potential is in the same scaling with that of a global $2$-design $\mathcal{F}^{(2)}_{\rm Haar}=2^{1-2N}\left( 2^{-2N} + 1 \right)$, in spite that the exact value of $\mathcal{F}^{(2)}_{\rm SEA}$ is larger than $\mathcal{F}^{(2)}_{\rm Haar}$. So we can regard the SEA as a better ansatz in the sense that it has better trainability while the expressibility is not sacrificed too much. From the simulation results in Fig.~\ref{fig:frame_potential}, we can see that the second frame potential of the SEA is pretty close to the optimal value.

\begin{figure}[t]
\centering
\includegraphics[width = 0.5\textwidth]{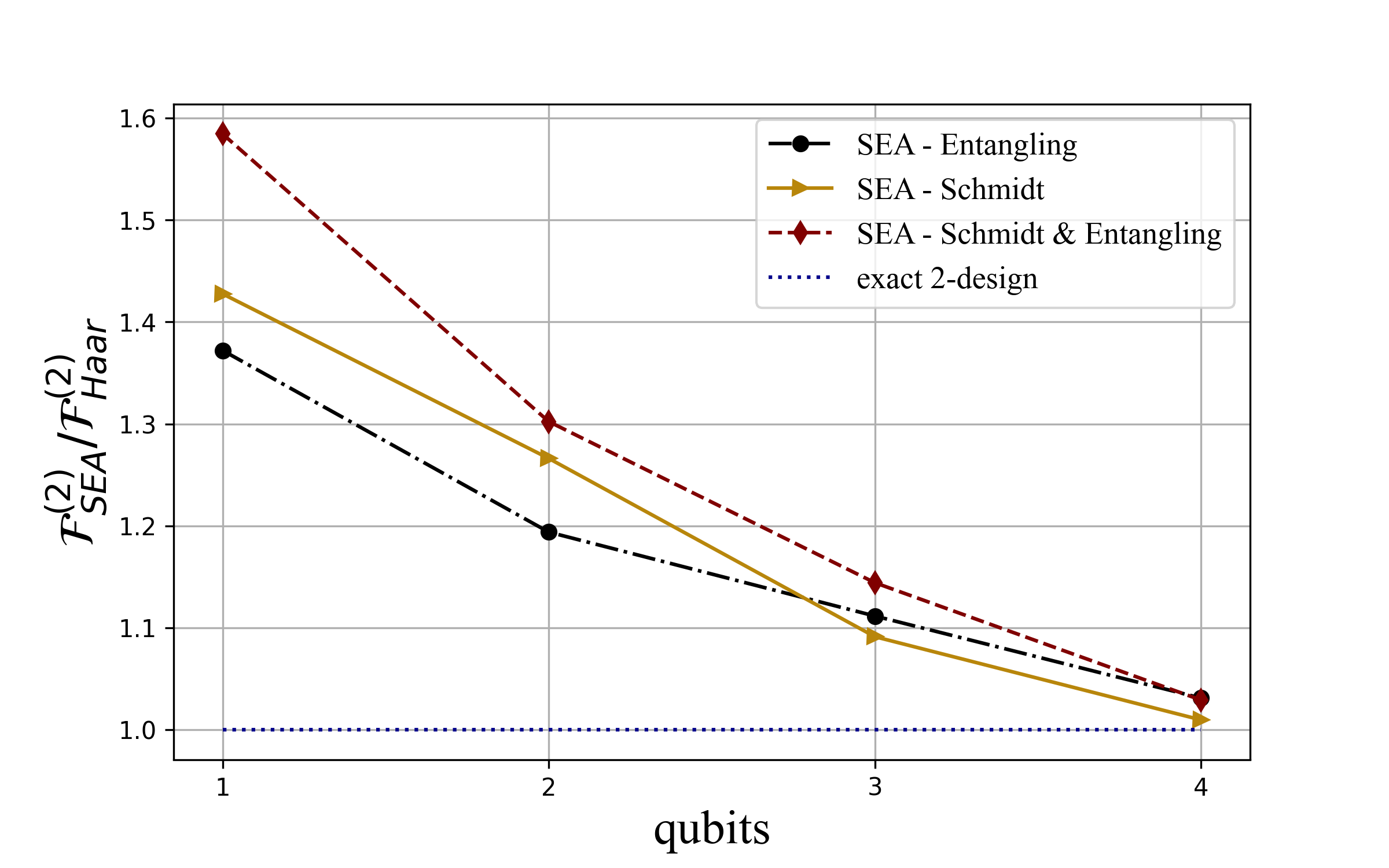}
\caption{\textbf{Second frame potentials of 8-qubit SEA with different Schmidt coefficient layers and entangling layers.}
	Yellow triangles, black dots, and red diamonds represent the frame potentials of 8-qubit SEA that change the number of qubits used in the Schmidt coefficient layer, entangling layer, and both two layers, respectively. Blue dotted line represents the frame potential of an exactly 2-design circuit. Each value is estimated using an ensemble of 5000 unitaries.}
\label{fig:frame_potential}
\end{figure}

\textbf{\textit{Numerical Simulation of Experiments.}}---\label{sec: application}
To verify the advantages of {\ansatz} on both efficiency and trainability, we investigate the performance as well as the magnitude of cost gradient of SEA and other ansatzes in estimating the ground state energy of chemistry and physics models by carrying out numerical simulations of experiments.
All simulations are performed using the Paddle Quantum~\cite{Paddle} toolkit on the PaddlePaddle Deep Learning Platform~\cite{Ma2019}. 


\paragraph*{Applying to chemistry and physics models.} 
To make sure the SEA can meet the aforementioned requirements for VQE, we adopt alternating layered ansatz (ALT)~\cite{cerezo2021cost}, which is a famous class in the hardware-efficient ansatz, for the Schmidt coefficient layer and the LBC layer in the following parts. An example of ALT is shown in Fig.~\ref{fig:alt}.


Here we use $14$-qubit dinitrogen (N$_2$) molecule and the $14$-qubit Heisenberg model as examples for numerical simulations of VQE. 
The structure of {\ansatz} is shown in Fig.~\ref{fig:TA-CNOT}, where $U_1$, $U_2$, and $U_3$ are all $7$-qubit ALTs with a depth of $30$. For different Hamiltonians, we could 
adjust the $K$ in Proposition~\ref{pro: truncation-general} to obtain SEAs with different entangling capabilities. Considering that the N$_2$ molecule and the Heisenberg model we choose are both weakly entangled, we use two kinds of {\ansatz} named {\ansatz}$_2$ and {\ansatz}$_3$ to learn these models' ground state energies. To be specific, {\ansatz}$_2$ is the $14$-qubit {\ansatz} of Schmidt coefficient layer as a $2$-qubit ALT and entangling layer as $2$ CNOTs, which means a composition of $2$ CNOTs controlled and targeted on the qubit-pairs $\{(0, 7), (1, 8)\}$. {\ansatz}$_3$ is the $14$-qubit {\ansatz} of Schmidt coefficient layer as a $3$-qubit ALT and entangling layer as $3$ CNOTs, which implement on the qubit-pairs $\{(i,7+i)\}_{i=0}^2$. As mentioned before, we can also set Schmidt coefficient layer as a subSEA. To verify this point, we use an SEA with Schmidt coefficient layer being the $7$-qubit subSEA, which is called SEA$\_$Sch, in our simulations.

For the purpose of studying the advantages of {\ansatz}, we compare a $14$-qubit {\ansatz} with a $14$-qubit ALT and a $14$-qubit random circuit~\cite{mcclean2018barren} in the task of VQE with a similar number of parameters by setting different depth for different ansatzes. The specific calculation methods of parameter numbers are given in the Sec.VII
of the Supplemental Material. The reason for choosing an ALT and a random circuit is that they both have a strong expressibility and they can be regarded as a $2$-design when their depth are large~\cite{mcclean2018barren,cerezo2021cost}. After setting up the ansatzes, we employ a stochastic gradient descent optimizer to iteratively update parameters for $400$ iterations.

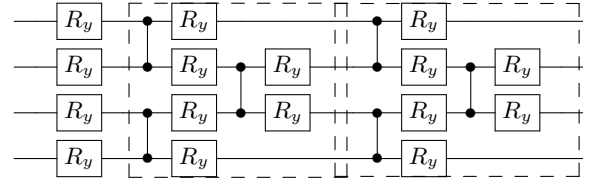
\begin{figure}
\centering
$$\Qcircuit @C=.9em @R=.4em{
	&\qw &\gate{R_y} &\qw &\ctrl{1} &\gate{R_y} &\qw      &\qw        &\qw &\qw        &\ctrl{1} &\gate{R_y} &\qw      &\qw         &\qw&\qw  \\
	&\qw &\gate{R_y} &\qw &\ctrl{0}   &\gate{R_y} &\ctrl{1} &\gate{R_y} &\qw &\qw        &\ctrl{0}   &\gate{R_y} &\ctrl{1} &\gate{R_y}  &\qw&\qw  \\
	&\qw &\gate{R_y} &\qw &\ctrl{1} &\gate{R_y} &\ctrl{0}   &\gate{R_y} &\qw &\qw        &\ctrl{1} &\gate{R_y} &\ctrl{0}   &\gate{R_y}  &\qw&\qw  \\
	&\qw &\gate{R_y} &\qw &\ctrl{0}  &\gate{R_y} &\qw &\qw   &\qw  &\qw &\ctrl{0}   &\gate{R_y} &\qw &\qw   &\qw&\qw 
	\gategroup{1}{5}{4}{9}{1.3em}{--} \gategroup{1}{10}{4}{15}{1.5em}{--}
}$$     
\caption{An example of a $4$-qubit ALT with $2$ layers. Each layer is composed of the rotation operator gate $R_y$ and control-Z gate.}
\label{fig:alt}
\end{figure}



\begin{figure}[t]
\centering
\includegraphics[width=0.5\textwidth]{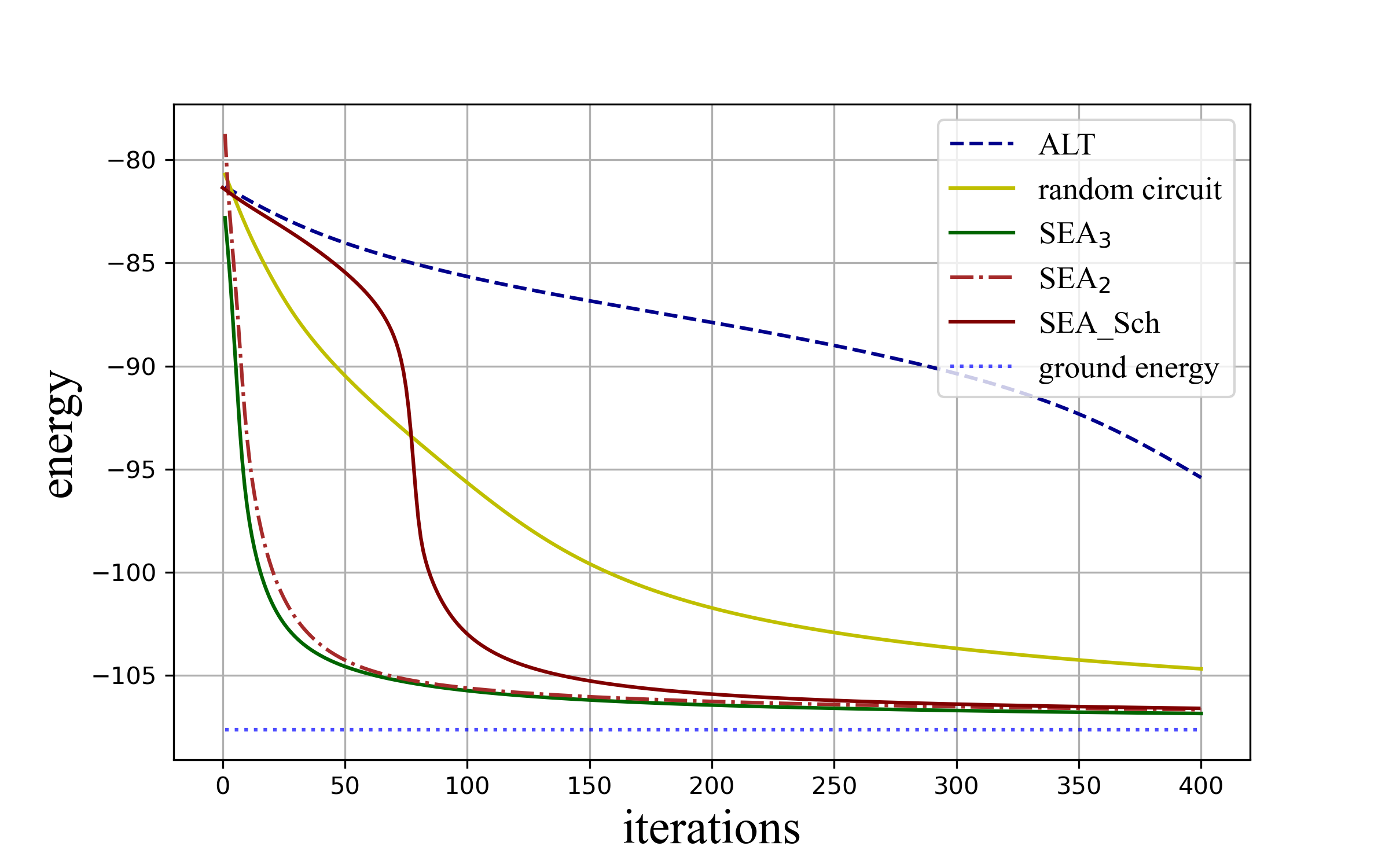}
\caption{\textbf{Numerical simulations of VQE on N$_2$ ($14$-qubit).}
	The blue dotted line is the theoretical ground energy of N$_2$, and the lines from top to bottom represent the experimental results of ALT, the random circuit, SEA of Schmidt coefficient layer
	as subSEA, {\ansatz} with $2$ CNOTs, and {\ansatz} with $3$ CNOTs, respectively. $j$($j=2, 3$) CNOTs means that we set a composition of $j \CNOT$s controlled and targeted on the qubit-pairs $\{(i, N+i)\}_{i=0}^{j-1}$.}
\label{fig:VQE-N2}
\end{figure}

\begin{figure}[t]
\centering
\includegraphics[width=0.5\textwidth]{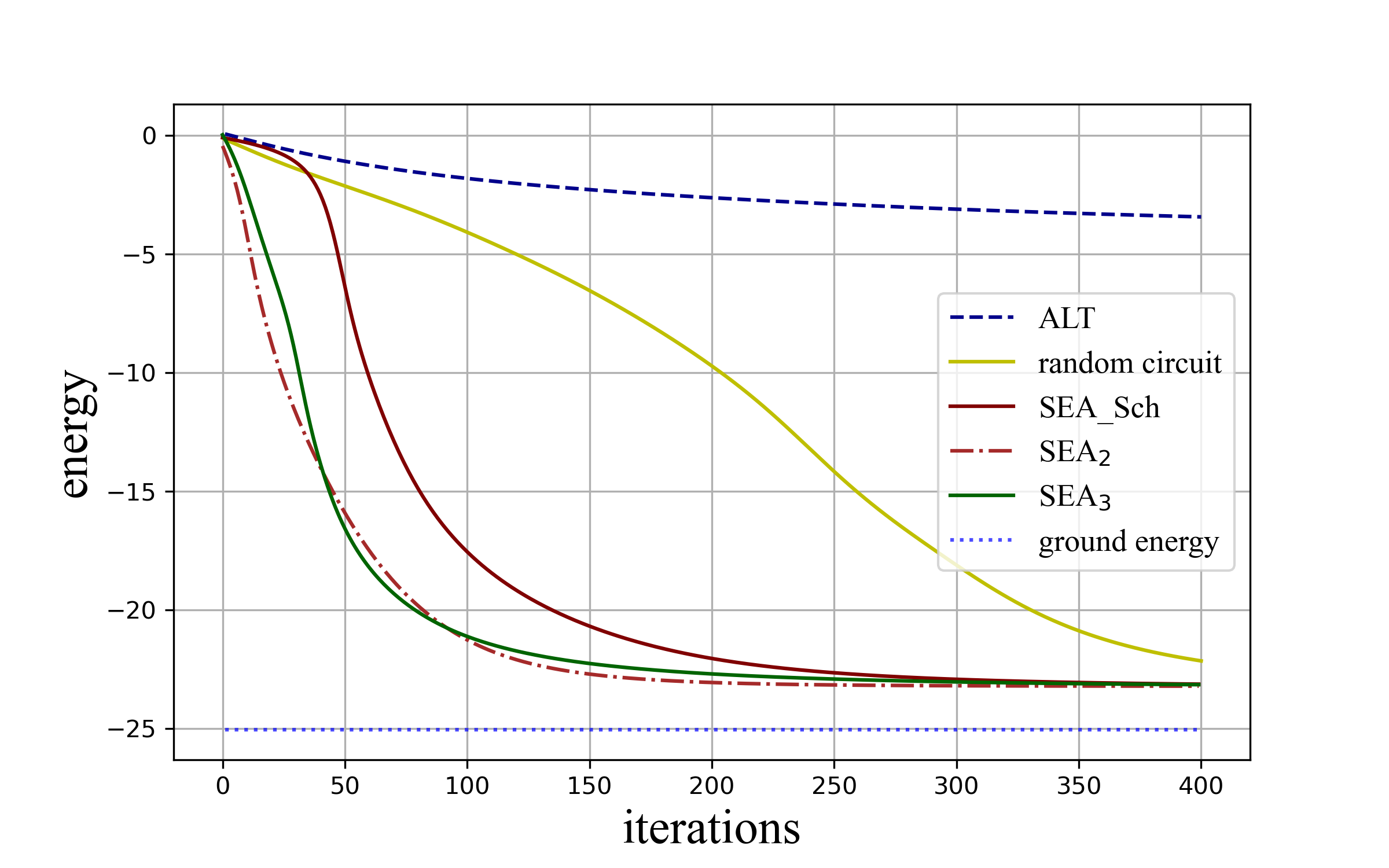}
\caption{\textbf{Numerical simulations of VQE on Heisenberg model ($14$-qubit).} The Hamiltonian is chosen to be the $14$-qubit $1$-dimensional spin-$1/2$ antiferromagnetic Heisenberg model with periodical boundary condition, i.e., $H=\sum_{i=1}^{14}\left( X_{i} X_{i+1} + Y_{i} Y_{i+1} + Z_{i} Z_{i+1} \right)$. The blue dotted line is the theoretical ground energy of the Hamiltonian $H$, and the other lines from top to bottom represent the results of the ALT, random circuit, SEA of Schmidt coefficient layer
	as subSEA, SEA with 3 CNOTs and SEA with 2 CNOTs, respectively.}
\label{fig:VQE_Heisenberg_14_sgd}
\end{figure}
The simulation results of N$_2$ molecule and the Heisenberg model are shown in Fig.~\ref{fig:VQE-N2} and Fig.~\ref{fig:VQE_Heisenberg_14_sgd}, respectively. Both figures illustrate the efficiency of SEA on VQE. We can observe that with a similar number of parameters, there is a visible gap in the accuracy of estimating the ground energy by different ansatzes. Both {\ansatz}$_2$ and {\ansatz}$_3$ converge rapidly and get an approximate value of ground energy while the results of the ALT and random circuit are hard to be optimized in $400$ iterations. Furthermore, it is also valid to replace the three parts of {\ansatz} with other structures. The SEA$\_$Sch in Fig.~\ref{fig:VQE-N2} and Fig.~\ref{fig:VQE_Heisenberg_14_sgd} both show that when we choose Schmidt coefficient layer as subSEA, this kind of SEA also has a good performance on VQE like SEA$_2$ or SEA$_3$. In addition to SEA$\_$Sch, the Sec.~VII
of the Supplemental Material presents the result of {\ansatz} whose Schmidt coefficient layer is $R_y(\bm{\theta})^{\otimes N}$. We can also see that this kind of {\ansatz} has better performance than ALT and random circuit, which further verifies the effectiveness of {\ansatz} in practical applications.

From the numerical experiment of VQE, we can see that {\ansatz} has an advantage in efficiency. Since we only use $2$ or $3$ CNOTs and adjust the structure of Schmidt coefficient layer to obtain good results, the experiment results also support Proposition~\ref{pro: truncation-general}.

\paragraph*{Gradient experiments.}
For both experiments on gradients with different numbers of qubits, we consider the $1$-dimensional spin-$1/2$ antiferromagnetic Heisenberg model where the Hamiltonian is $H=\sum_{i=1}^{2N}( X_{i} X_{i+1} + Y_{i} Y_{i+1} + Z_{i} Z_{i+1})$ with periodic boundary condition. We use different ansatzes and calculate the gradient of the cost function $C(\bm{\theta}) = \bra{0}^{\otimes 2N}U(\bm{\theta})^{\dagger}HU(\bm{\theta})\ket{0}^{\otimes 2N}$ with respect to each parameter, where $U(\bm{\theta})$ represents the unitary of the ansatz.

We first conduct experiments with the ansatzes used in the previous numerical simulations of VQE. The SEAs we chose are SEA$_1$ and SEA$_{\lfloor N/2\rfloor}$, which still adopt ALT as the Schmidt coefficient layer and the LBC layer.
For each $N$, ALT and random circuit with similar numbers of parameters as SEAs are selected for comparison.

\begin{figure}[t]
\centering
\includegraphics[width=0.5\textwidth]{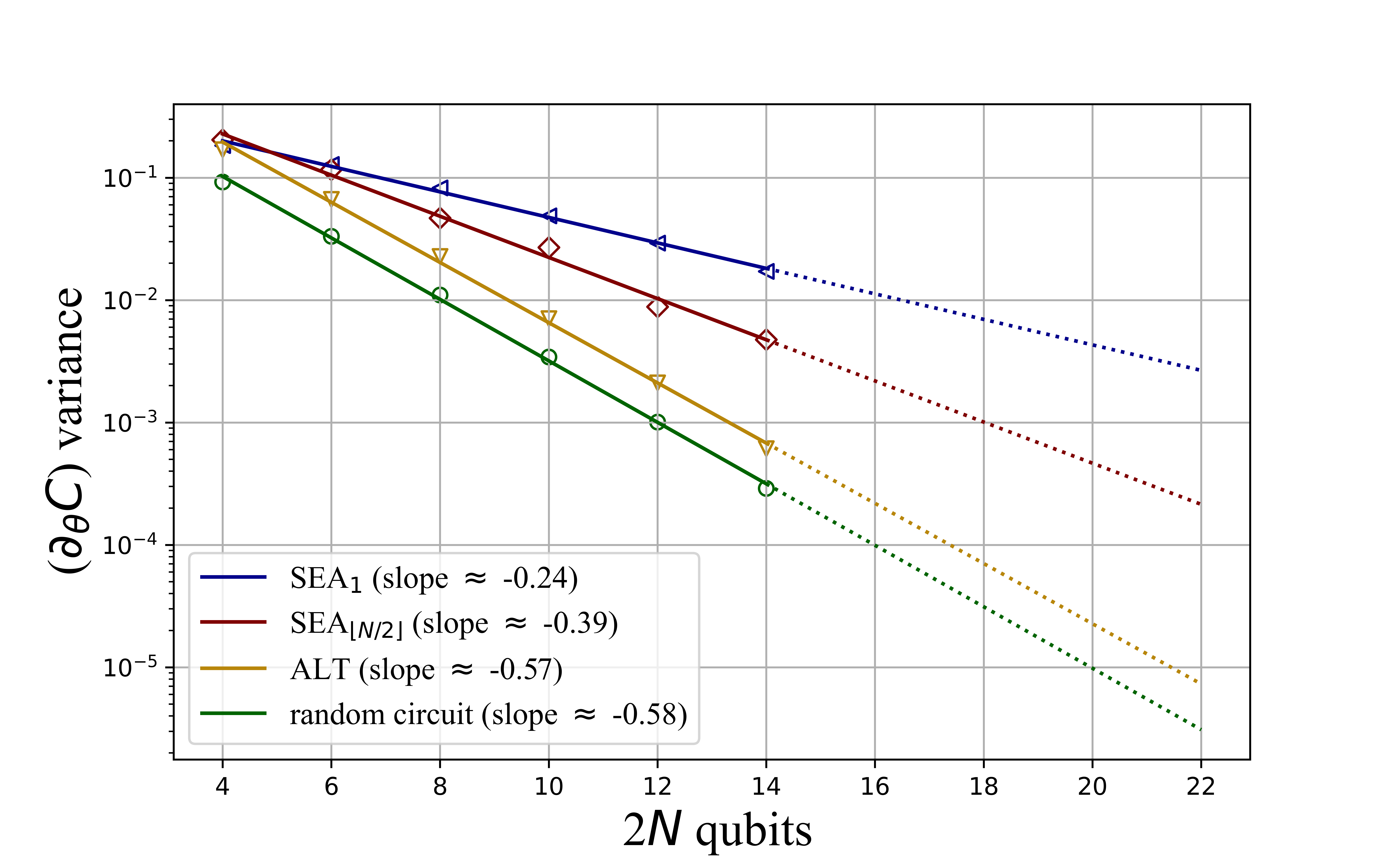}
\caption{\textbf{Comparison of the scaling of the average variance of gradients between different ansatzes on the Heisenberg model}. This shows the semi-log plot of the average value of the variance of each parameters in the circuit. The depth of each part of SEA is 30 and we ensure different ansatzes have similar numbers of parameters by setting different depth. The solid part of the fitted lines represents the range we experiment with, while the dashed part represents the expected performance on a larger range.}
\label{Fig: BP-sub.1}
\end{figure}

The average variances of parameters of different ansatzes are compared in Fig.~\ref{Fig: BP-sub.1}.
We also extract the variance of the largest partial derivative in each sample, which is shown in the Sec.~VII
of the Supplemental Material. It is clearly demonstrated that all these variances decay exponentially as the number of qubits increases. However, the slopes of the red and the blue lines are nearly half of the slopes of the other lines, which means the {\ansatz} has a much lower rate of decay of variance. 


To better illustrate how our structure increases the scaling of variance, we further inspect the gradient with respect to the parameter in different parts of SEA. As shown in Fig.~\ref{fig:RU_QNN}, we look into the gradients with respect to the parameters of the $R_y$ gates in the Schmidt coefficient layer and the LBC layer.
As a comparison, we also consider a single parameter $R_y$ gate in the middle of a circuit that forms a 2-design with sufficient depth.
The two parts split by the gate are both local 2-designs and
are represented by Haar random unitaries in our experiments. Fig.~\ref{fig:BP_specific} summarizes the results of the variance of parameters located in the Schmidt coefficient layer and the LBC layer, respectively. 

\begin{figure}[t]
\centering
\includegraphics[width=0.45\textwidth]{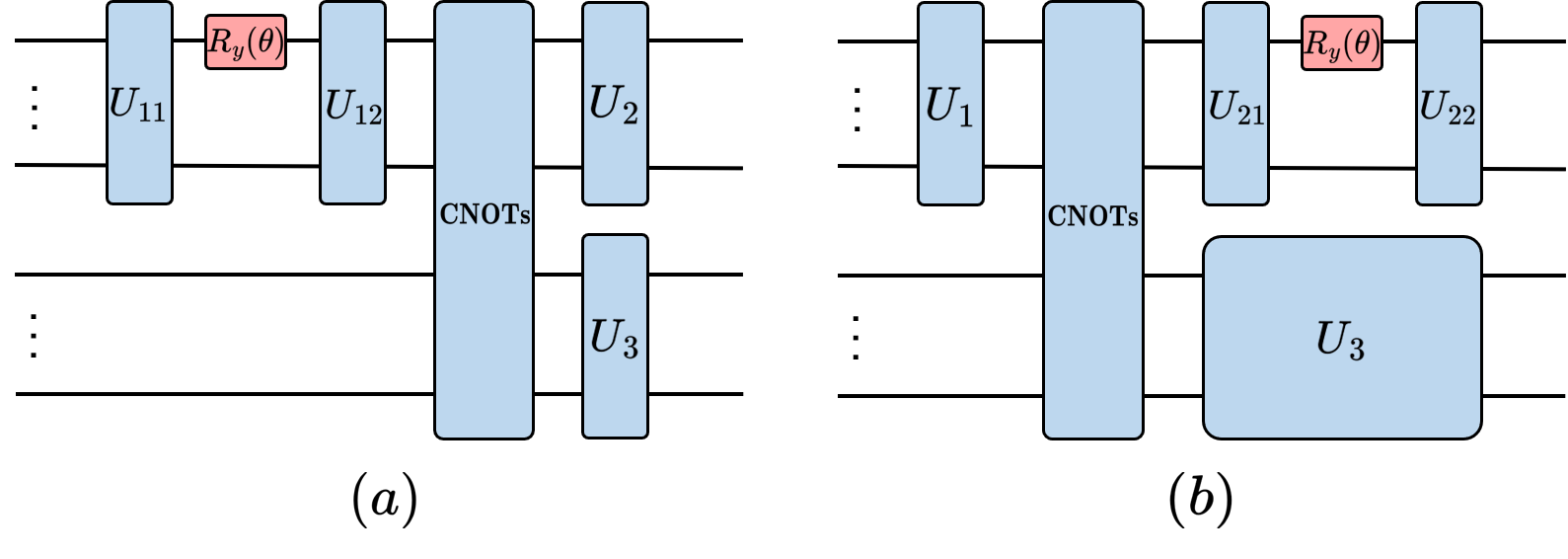}
\caption{\textbf{The structures of SEA used for studying the scaling of variance of general circuits} Panel (a) shows the $R_y$ gate located in Schmidt coefficient layer. Panel (b) shows the $R_y$ gate located in the LBC layer. Parameter $\theta \in [0,2\pi)$ is uniformly chosen and $U_1,U_{11},U_{12},U_2,U_{21},U_{22},U_3$ are Haar randomly chosen from unitary group in each round of sampling. $CNOTs$ is composed by a number of CNOTs.}
\label{fig:RU_QNN}
\end{figure}
\begin{figure}[t]
\centering
\includegraphics[width=0.5\textwidth]{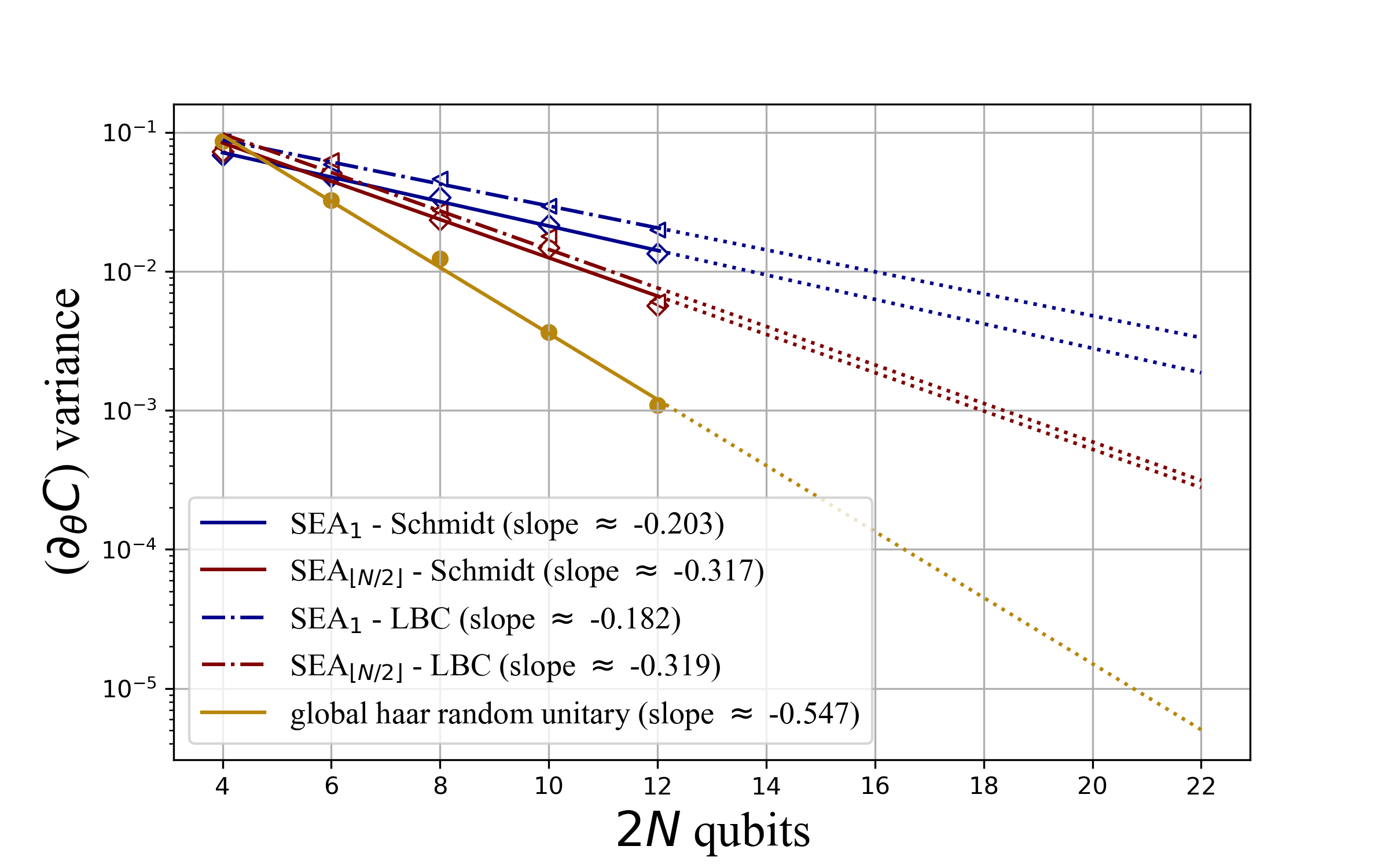}
\caption{\textbf{Comparison of the scaling of variance between general structure of SEA and general random circuit.} The SEA$_1$-Schmidt and the SEA$_{\lfloor N/2 \rfloor}$-Schmidt show the semi-log plot of the variance of gradient with respect to the parameter in Schmidt coefficient layer. The SEA$_1$-LBC and the SEA$_{\lfloor N/2 \rfloor}$-LBC show the semi-log plot of the variance of gradient with respect to the parameter in LBC layer. The markers on the fitted lines represent the results from our experiments. The dotted part of the lines represents the expected performance on a larger range.}
\label{fig:BP_specific}
\end{figure}

From Fig.~\ref{fig:BP_specific}, we see that no matter where this parameter is located in {\ansatz}, the variance of its gradient has a decay speed lower than the decay speed of variance of the parameter in a sufficiently random circuit that forms a 2-design. This fact indicates that the variance of {\ansatz} has significant improvement compared to using an ansatz that forms a 2-design, indicating an evident mitigation of barren plateaus. It is also worth noting that the fewer CNOTs in the entangling layer, the larger variance it will have, which is in line with Proposition~\ref{prop:gradient-variance}.
Specifically, in our experiments, if the scale of the problem is larger than 12 qubits, the variance of the gradients in a 2-design circuit will be less than $10^{-3}$. However, this bound of scale can be increased to 18 qubits using {\ansatz} with $\lfloor \frac{N}{2}\rfloor$ CNOTs according to our experiment results. With SEA of weaker expressibility, this bound can be further increased. 
In other words, we can notably increase the scale of problems with the same limited scaling of gradients by using {\ansatz} instead of a 2-design circuit. Therefore, {\ansatz} is more likely to scale to larger systems on VQE tasks because of its improvement in trainability by mitigating barren plateaus.

\textbf{\textit{Concluding remarks.}}---\label{sec: conclusion}\label{sec: conclusion}
The barren plateaus phenomenon that comes with the strong expressibility of an ansatz severely limits the trainability of the VQE. The main contribution of this work is proposing an ansatz structure {\ansatz} to mitigate the barren plateaus in optimizing parameterized circuits for quantum dynamics simulations. With the idea of Schmidt decomposition, we invent the SEA by removing the redundant expressibility in universal PQCs. We in particular have provided an explicit construction of {\ansatz} that has enough expressibility to generate arbitrary pure quantum states. Moreover, SEA does not form a unitary 2-design, and by tuning the number of CNOT gates, it can efficiently mitigate barren plateaus. From the perspective of frame potential, we prove that SEA has high trainability without losing too much expressibility, which allows a trade-off between trainability and expressibility. Numerical simulations show that {\ansatz} can be applied to approximate the ground state of the N$_2$ molecule and the Heisenberg model, whose numerical results confirm the effectiveness of {\ansatz}. We further compare the scaling of variance in partial derivatives between different ansatzes and establish that {\ansatz}
can obtain a quadratically lower 
rate of decay of variance, implying that {\ansatz} can considerably mitigate the barren plateaus of VQE.

We note that {\ansatz} is also effective in principle for other tasks with the sole purpose of generating a pure state, e.g., combinatorial optimization~\cite{Farhi2014} and entanglement detection~\cite{wang2022detecting}. It is natural to explore how to extend our structure by considering multipartite systems and other possible designs. We also remark that there are other proposals for dealing with barren plateaus~\cite{Grimsley2022,Tang2021,Skolik2021,Volkoff2020}, and for using meta-learning to improve efficiency~\cite{Verdon2019b,Wilson2021,Stollenwerk2022,Chen2021a}. The combination between SEA and other approaches such as adaptive methods may be worth a future study. Moreover, with recent works exploiting the symmetry of VQAs~\cite{Zheng2021,Meyer2022,Larocca2022}, it is also worth studying how to improve the performance of SEA from this perspective.

\textbf{\textit{Acknowledgements.}} 
X. L. and G. L. contributed equally to this work. 
We would like to thank Runyao Duan, Yin Mo, Chengkai Zhu, and Youle Wang for helpful discussions. 
Part of this work was done when X. L., G. L., J. H. and H. Z. were interns at Baidu Research.



\bibliographystyle{apsrev4-1}
\bibliography{smallbib}

\clearpage
\vspace{2cm}
\onecolumngrid
\vspace{2cm}

\begin{center}
{\textbf{\large Supplemental Material for \\ Mitigating barren plateaus of variational quantum eigensolvers}}
\end{center}

\renewcommand{\theequation}{S\arabic{equation}}
\renewcommand{\theproposition}{S\arabic{proposition}}
\renewcommand{\theHproposition}{S\arabic{proposition}}
\renewcommand{\thedefinition}{S\arabic{definition}}
\renewcommand{\thefigure}{S\arabic{figure}}
\renewcommand{\thetable}{S\arabic{table}}
\setcounter{equation}{0}
\setcounter{section}{0}
\setcounter{proposition}{0}
\setcounter{figure}{0}

\section{Preliminaries}
Here we present some supplemental lemmas and define some notations to be used throughout the proof.
We write a target quantum state, which is a pure state to be approximated by a PQC, as $|\phi\rangle$ and let the actual circuit output be $|\psi(\bm{\theta})\rangle$. 
We may simply write $|\psi\rangle$ if the circuit parameters $\bm{\theta}$ are unimportant in the context.
Moreover, let $\sfid=|\langle\psi|\phi\rangle|^2$ be fidelity.
Unless otherwise specified, $\{\ket{k}\}_{k=0}^{2^N-1}$ are the $N$-qubit computational bases.
We denote $\BCNOT\equiv\prod_{i=0}^{N-i}\CNOT(i, N+i)$ which is a composition of $N \CNOT$s controlled and targeted on the pair $(i, N+i)$ for all $i=0,\dots,N-1$. Let $\bin{i}{}$ be a vector for the binary expansion of $i$, and therefore $\bin{i}{j}$, where $0\leq j<\lceil\log i\rceil$, is the $j$-th leading digit of the expansion.




\subsection{$t$-design and integration over the unitary group}\label{Supp:unitary-design}
We firstly recall the definition of a $t$-design. Assume $\mathcal{U} = \{U_k\}_{k=1}^K$ is a finite of unitary operator on $\mathbb{C}^D$, and $P_{t,t}(U)$ is a polynomial of degree at most $t$ in the matrix elements of $U$ and at most $t$ in those of $U^{\dagger}$. Then, we define that $\mathcal{U}$ is a $t$-design if for every polynomial $P_{t,t}(U)$ we have
\begin{align}
\frac{1}{K}\sum_{k=1}^K P_{t,t}(U_k) = \int P_{t,t}(U)d\eta(U), \label{def:t-design}
\end{align}
where the integrals with respect to the Haar measure over the unitary group. Particularly, when $t=2$, the definition is equivalent to the following definition \cite{dankert2009exact}.

\renewcommand{\theproposition}{S\arabic{definition}}
\renewcommand{\theproposition}{S\arabic{proposition}}

\begin{definition}
$\{U_k\}_{k=1}^K$ forms a unitary $2$-design if and only if for any linear operators $C,D,\rho \in L(\mathbb{C}^D)$, we have
\begin{align} 
	\frac{1}{K}\sum_{k=1}^K U_k^{\dagger}CU_ k\rho U_k^{\dagger}DU_k 
	&=\quad\int_{\mathcal{U}(d)}U^{\dagger}CU\rho U^{\dagger}DUd\eta(U). \label{Appendix: def-2-design}
\end{align}
\end{definition}

Based on this definition, we can know whether a unitary group is a 2-design by comparing the both sides of Eq.~\eqref{Appendix: def-2-design}. Fortunately, according to the Schur's lemma \cite{feit1982representation} , the RHS of Eq.~\eqref{eq:Appendix-lem:2-design-U} has a closed form which is shown in lemma \ref{Appendix-lem:2-design-U}. The proof of lemma \ref{Appendix-lem:2-design-U} can be seen in \cite{emerson2005scalable}.

\begin{lemma}
\label{Appendix-lem:2-design-U}
Let $\{U_k\}_{k=1}^K$ forms a unitary $t$-design with $t\geq 2$, for any linear operators $C,D,\rho \in L(\mathbb{C}^D)$. We have
\begin{align} 
	\quad\int_{\mathcal{U}(d)}U^{\dagger}CU\rho U^{\dagger}DUd\eta(U)
	= \frac{\tr[CD]\tr[\rho]}{d}\frac{I}{d}+(\frac{d\tr[C]\tr[D]-\tr[CD]}{d(d^2-1)})(\rho-\tr[\rho]\frac{I}{d}).\label{eq:Appendix-lem:2-design-U}
\end{align}
\end{lemma}

Furthermore, we present the following lemma so that we can solve the problem with bipartite system.

\begin{lemma}\label{lem:2-design-UxU}
For any bipartite state $\rho_{AB}$ ($d_A=d_B=d$), and arbitrary linear operators $C,D \in L(\mathbb{C}^{d^2})$, we have
\begin{align}
	\int_{\mathcal{U}(d)\otimes \mathcal{U}(d)} d \eta(U) U^{\dagger} C U \rho U^{\dagger} D U = t_0 \rho + t_1 \frac{\rho_A\ox I_B}{d} + t_2\frac{I_A\ox\rho_B}{d} + t_3 I_{AB}\tr(\rho_{AB}),\label{eq: u1-tensor-u2}
\end{align}
where $\rho_A = \tr_B[ \rho_{AB}]$, $\rho_B = \tr_A [\rho_{AB}]$,and $\{t_j\}_{j=0}^3$ can be computed from the following linear system of equations
\begin{align}
	\tr(CD) & = t_0d^2 + t_1d^2 +t_2d^2 + t_3d^4, \label{eq:U1U2-part1}   \\
	\tr(C_A D_A) & = t_0d^3+t_1d+t_2d^3+t_3d^3,\label{eq:U1U2-part2} \\
	\tr(C_B D_B) & = t_0d^3+t_1d^3+t_2d+t_3d^3,\label{eq:U1U2-part3} \\
	\tr(C)\tr(D) & = t_0d^4 + t_1d^2 + t_2d^2+t_3d^2,\label{eq:U1U2-part4}
\end{align}
that is,
\begin{equation}
	\left[ {\begin{array}{*{20}{c}}
			{t_0}\\
			{t_1}\\
			t_2\\
			t_3
	\end{array}} \right] =\frac{1}{\left(d^{2}-1\right)^{2}}\left[ {\begin{array}{*{20}{c}}
			\frac{\tr[CD]}{d^2}-\frac{\tr[C_AD_A]}{d}-\frac{\tr[C_BD_B]}{d}+\tr[C]\tr[D]\\
			-\tr[CD]+\frac{\tr[C_AD_A]}{d}+d\tr[C_BD_B]-\tr[C]\tr[D] \\
			-\tr[CD]+d\tr[C_AD_A]+\frac{\tr[C_BD_B]}{d}-\tr[C]\tr[D]  \\
			\tr[CD]-\frac{\tr[C_AD_A]}{d}-\frac{\tr[C_BD_B]}{d}+\frac{\tr[C]\tr[D]}{d^2}
	\end{array}} \right].
\end{equation}
\end{lemma}

\begin{proof}
Similar to the proof of Lemma~\ref{Appendix-lem:2-design-U}, Schur's lemma implies that the LHS of Eq.~\eqref{eq: u1-tensor-u2} has an explicitly expression, which is denoted as follows:

\begin{align}
	\int_{\mathcal{U}(d)\otimes \mathcal{U}(d)} d \eta(U) U^{\dagger} C U \rho U^{\dagger} D U = t_0 \rho + t_1 \frac{\rho_A\ox I_B}{d} + t_2\frac{I_A\ox\rho_B}{d} + t_3 I_{AB}\tr(\rho_{AB}).
\end{align}
For simplicity, we define the following two functions:
\begin{align}
	&f^{(1)}(\rho)\equiv t_0 \rho + t_1 \rho_A\ox I_B/d + t_2I_A\ox\rho_B/d + t_3 I_{AB}\tr(\rho_{AB}),\\
	&f^{(2)}(\rho)\equiv  \int_{\mathcal{U}(d)\otimes \mathcal{U}(d)} d \eta(U) U^{\dagger} C U \rho U^{\dagger} D U ,
\end{align}
where $f^{(1)} = f^{(2)}$.

Then we only need to consider the coefficients $t_0,t_1,t_2,t_3$. To calculate the coefficients, we defined the following four operators:
\begin{align}
	&T_1(f)\equiv\sum_{i,j}\bra{ij}f(I)\ket{ij},\\
	&T_2(f)\equiv\sum_{i,j,k}\bra{ik}f(\ketbra{i}{j}_A\otimes I_B)\ket{jk},\\
	&T_3(f)\equiv\sum_{i,k,l}\bra{ik}f(I_A \otimes \ketbra{k}{l}_B)\ket{il},\\
	&T_4(f)\equiv\sum_{i,j,k,l}\bra{ij}f(\ketbra{ij}{kl})\ket{kl}.
\end{align}

For Eq.~\eqref{eq:U1U2-part1}, we have
\begin{align}
	T_1(f^{(1)})&=t_0\sum_{i,j}\bra{ij}I\ket{ij}+t_1\sum_{i,j}\bra{ij}I_A\otimes I_B\ket{ij}+t_2\sum_{i,j}\bra{ij}I_A\otimes I_B\ket{ij}+t_3d^4\\
	&=t_0d^2 + t_1d^2 +t_2d^2 + t_3d^4,\\
	T_1(f^{(2)})&=\sum_{i,j}\bra{ij}\int_{\mathcal{U}(d)\otimes \mathcal{U}(d)} d \eta(U) U^{\dagger} C U I U^{\dagger} D U\ket{ij}\\
	&=\tr(CD).
\end{align}
Since $T_1(f^{(1)})=T_1(f^{(2)})$, then $\tr(CD) = t_0d^2 + t_1d^2 +t_2d^2 + t_3d^4$.

For Eq.~\eqref{eq:U1U2-part3}, we have
\begin{align}
	T_2(f^{(1)})&=t_0\sum_{i,j,k}\bra{ik}\ketbra{i}{j}\otimes I_B\ket{jk}+t_1\sum_{i,j,k}\bra{ik}\ketbra{i}{j}\otimes I_B\ket{jk}\nonumber\\
	&\quad +\frac{t_2}{d}\sum_{i,j,k}\bra{ik}I_A\otimes\tr[\ketbra{i}{j}]I_B\ket{jk}+t_3\sum_{i,j,k}\bra{ik}I_{AB}\tr[\ketbra{i}{j}\otimes I_{B}]\ket{jk}\\
	&=t_0d^3 + t_1d^3 +t_2d + t_3d^3,\\
	T_2(f^{(2)})&=\sum_{i,j,k}\bra{ik}\int_{\mathcal{U}(d)\otimes \mathcal{U}(d)} d \eta(U) U^{\dagger} C U (\ketbra{i}{j}_A\otimes I_B) U^{\dagger} D U \ket{jk}\\
	&=\sum_{i,j,k}\int_{\mathcal{U}(d)\otimes \mathcal{U}(d)}d \eta(U) \bra{ik}  U^{\dagger} C U (\ketbra{i}{j}_A\otimes I_B) U^{\dagger} D U \ket{jk} \\
	&=\tr(C_B D_B).
\end{align}
Since $T_2(f^{(1)})=T_2(f^{(2)})$, then $\tr(C_B D_B) = t_0d^3+t_1d^3+t_2d+t_3d^3$.

Similarly, we have
\begin{align}
	T_3(f^{(1)})&=t_0d^3+t_1d+t_2d^3+t_3d^3,\\
	T_3(f^{(2)})&=\tr(C_A D_A),\\
	T_4(f^{(1)})&=t_0d^4+t_1d^2+t_2d^2+t_3d^2,\\
	T_4(f^{(2)})&=\tr(C)\tr(D).\\
\end{align}
Since
\begin{align}
	T_3(f^{(1)})&=T_3(f^{(2)}),\\
	T_4(f^{(1)})&=T_4(f^{(2)}),
\end{align}
then
\begin{align}
	\tr(C_A D_A) & = t_0d^3+t_1d+t_2d^3+t_3d^3,\\
	\tr(C)\tr(D) & = t_0d^4 + t_1d^2 + t_2d^2+t_3d^2.
\end{align}
Up to now, we have proved Eq.~\eqref{eq:U1U2-part1} to Eq.~\eqref{eq:U1U2-part4}.

Let
\begin{equation}
	R = \left(\begin{array}{llll}
		d^{2} & d^{2} & d^{2} & d^{4} \\
		d^{3} & d^{1} & d^{3} & d^{3} \\
		d^{3} & d^{3} & d^{1} & d^{3} \\
		d^{4} & d^{2} & d^{2} & d^{2}
	\end{array}\right).
\end{equation}
Then
\begin{equation}
	R^{-1} = \frac{1}{\left(d^{2}-1\right)^{2}}\left(\begin{array}{cccc}
		\frac{1}{d^{2}} & -\frac{1}{d} & -\frac{1}{d} & 1 \\
		-1 & \frac{1}{d} & d & -1 \\
		-1 & d & \frac{1}{d} & -1 \\
		1 & -\frac{1}{d} & -\frac{1}{d} & \frac{1}{d^{2}}
	\end{array}\right).
\end{equation}
Hence, we have 
\begin{equation}
	\left[ {\begin{array}{*{20}{c}}
			{t_0}\\
			{t_1}\\
			t_2\\
			t_3
	\end{array}} \right] = R^{-1}\left[ {\begin{array}{*{20}{c}}
			\tr(CD)\\
			\tr(C_A D_A) \\
			\tr(C_B D_B) \\
			\tr(C)\tr(D)
	\end{array}} \right]=\frac{1}{\left(d^{2}-1\right)^{2}}\left[ {\begin{array}{*{20}{c}}
			\frac{\tr[CD])}{d^2}-\frac{\tr[C_AD_A]}{d}-\frac{\tr[C_BD_B]}{d}+\tr[C]\tr[D]\\
			-\tr[CD]+\frac{\tr[C_AD_A]}{d}+d\tr[C_BD_B]-\tr[C]\tr[D] \\
			-\tr[CD]+d\tr[C_AD_A]+\frac{\tr[C_BD_B]}{d}-\tr[C]\tr[D]  \\
			\tr[CD]-\frac{\tr[C_AD_A]}{d}-\frac{\tr[C_BD_B]}{d}+\frac{\tr[C]\tr[D]}{d^2}
	\end{array}} \right],
\end{equation}

\begin{equation}
	\left[ {\begin{array}{*{20}{c}}
			{t_0}\\
			{t_1}\\
			t_2\\
			t_3
	\end{array}} \right] =\frac{1}{\left(d^{2}-1\right)^{2}}\left[ {\begin{array}{*{20}{c}}
			\frac{\tr[CD]}{d^2}-\frac{\tr[C_AD_A]}{d}-\frac{\tr[C_BD_B]}{d}+\tr[C]\tr[D]\\
			-\tr[CD]+\frac{\tr[C_AD_A]}{d}+d\tr[C_BD_B]-\tr[C]\tr[D] \\
			-\tr[CD]+d\tr[C_AD_A]+\frac{\tr[C_BD_B]}{d}-\tr[C]\tr[D]  \\
			\tr[CD]-\frac{\tr[C_AD_A]}{d}-\frac{\tr[C_BD_B]}{d}+\frac{\tr[C]\tr[D]}{d^2}
	\end{array}} \right].
\end{equation}

\end{proof}

\section{The effectiveness of SEA on VQE}
\label{appendix:effectiveness-of-vqe}
In this section, we present a proof of the effectiveness of SEA on VQE. We first give the Proposition~\ref{pro:state} to illustrate that SEA can generate an arbitrary pure state. Then in Proposition~\ref{pro:VQE} we demonstrate that the SEA in Proposition~\ref{pro:state} can find the ground state energy of a Hamiltonian, which implies the effectiveness of SEA on VQE.

\renewcommand{\theproposition}{S\arabic{proposition}}

\begin{proposition}
\label{pro:state}
If $U_1$ can generate an arbitrary $N$-qubit pure state, then given any $2N$-qubit pure state $\ket{\phi}$, a $2N$-qubit SEA can generate $\ket{\phi}$ with a certain set of parameters $\bm{\hat{\theta}} =  \{\hat{\bm{\theta}}_1,\hat{\bm{\theta}}_2,\hat{\bm{\theta}}_3\}$, that is,
\begin{align}
	S(\bm{\hat{\theta}})\ket{0}^{\otimes 2N} =  \ket{\phi}.
\end{align}
\end{proposition}

\begin{proof}


Beginning with Schmidt decomposition, we can write the target state $\ket{\phi}$ as
\begin{align}
	\ket{\phi}=(\hat{U}_2\otimes \hat{U}_3)\sum_{k=0}^{2^N-1}\lambda_k\ket{k}_A\ket{k}_B, 
\end{align}
where $A$ and $B$ are two $N$-qubit subsystems, $\hat{U}_2$ and $\hat{U}_3$ are two unitary operators acting on these two subsystems, $\{\ket{k}\}_{k=0}^{2^N-1}$ are the $N$-qubit computational bases, and $\{\lambda_k\}_{k=0}^{2^N-1}$ are Schmidt coefficients.

Because $U_1(\bm{\theta_1})$ can generate an arbitrary $N$-qubit pure state, and $U_2(\bm{\theta_2}), U_3(\bm{\theta_3})$ are universal, we can choose a certain set of parameters $\bm{\hat{\theta}} =  \{\hat{\bm{\theta}}_1,\hat{\bm{\theta}}_2,\hat{\bm{\theta}}_3\}$ such that
\begin{align}
	U_1(\hat{\bm{\theta}}_1)\ket{0}^{\otimes N} &=\sum_{k=0}^{2^N-1} \lambda_k\ket{k},\\
	U_2(\hat{\bm{\theta}}_2) &= \hat{U}_2,\\
	U_3(\hat{\bm{\theta}}_3) &= \hat{U}_3.
\end{align}

Combining with $V$, which is a composition of $N \CNOT$s controlled and targeted on the qubit-pairs $\{(i, N+1)\}_{i=0}^{N-1}$, there has 
\begin{align}
	V (U_1(\bm{\theta_1})\otimes I)\ket{0}_A^{\otimes N}\ket{0}_B^{\otimes N}
	= & \sum_{k=0}^{2^N-1} \lambda_k V \ket{k}_A\ket{0}_B^{\otimes N} \\
	= & \sum_{k=0}^{2^N-1} \lambda_k\ket{k}_A\ket{k}_B.
\end{align}

Therefore, we have
\begin{align}
	S(\hat{\bm{\theta}})\ket{0}^{\otimes 2N}
	&=(U_2(\hat{\bm{\theta}}_2)\otimes U_3(\hat{\bm{\theta}}_3))V(U_1(\hat{\bm{\theta}}_1)\otimes I)\ket{0}^{\otimes N}\ket{0}^{\otimes N}\\
	&=(\hat{U}_2\otimes \hat{U}_3)\sum_{k=0}^{2^N-1} \lambda_k\ket{k}_A\ket{k}_B. \\
	&=\ket{\phi}.
\end{align}

\end{proof}

\begin{proposition}
\label{pro:VQE}
For any $2N$-qubit Hamiltonian $H$, it holds that 
\begin{align}
	\min_{S} \bra{0}^{\otimes 2N}S^{\dagger}HS\ket{0}^{\otimes 2N} = E_0,
\end{align}
where $E_0$ is the ground state energy of $H$
and the optimization is over all unitaries reachable by {\ansatz} with $U_1$ having universal wavefunction expressibility.
\end{proposition}

\begin{proof}

A given $2N$-qubit Hamiltonian $H$ can be written as follows according to spectral decomposition:
\begin{align}
	H =\sum_{i=0}^m E_iP_i, 
\end{align}
where $\{E_i\}_{i=0}^{m}$ are eigenvalues of $H$ such that $E_0 < E_1 < \cdots < E_{m}$ and $P_i$ is a projector onto the eigenspace $V_i$ corresponding to the eigenvalue $E_i$.
Then given an arbitrary pure state $\ket{\psi} \in V_0$, we have
\begin{align}
	H \ket{\psi} = E_0\ket{\psi}.  
\end{align}
Note that when the optimization is over all unitaries reachable by {\ansatz} with $U_1$ having universal wavefunction expressibility, there exists an $S$ such that $\ket \psi = S\ket{0}^{\otimes 2N}$. 
Therefore,
\begin{align}
	\bra{0}^{\otimes 2N}S^{\dagger}HS\ket{0}^{\otimes 2N} &= \bra{\psi}H\ket{\psi}\\
	&= E_0\braket{\psi}{\psi}\\
	&= E_0.
\end{align}
Since it is trivial that $ \min_{S} \bra{0}^{\otimes 2N}S^{\dagger}HS\ket{0}^{\otimes 2N} \geq E_0$, we have
\begin{align}
	\min_{S} \bra{0}^{\otimes 2N}S^{\dagger}HS\ket{0}^{\otimes 2N} = E_0.
\end{align}

\end{proof}

\section{Proof of Proposition 1}
\label{appendix:proof-of-TA-state }

In this section, we give a detailed proof of Proposition~\ref{pro:truncation}. We first give the following lemma that is used in the proof.

\renewcommand{\theproposition}{S\arabic{proposition}}
\begin{lemma}\label{lemma:average}
For any descending sequence $\{x_i\}_{i=1}^{n}(x_i\geq x_{i+1})$, and any $n\geq N\geq M\geq 1$ ($n, N, M \in \mathbb{Z}$), the following is true:
\begin{align}
	\frac{1}{M}\sum_{i=1}^Mx_i\geq \frac{1}{N}\sum_{i=1}^Nx_i.
\end{align}
\end{lemma}

\begin{proof}
Because $N\geq M\geq 1$ and $x_i\geq x_{i+1}$, we have the following formula:
\begin{align}
	&\quad\frac{1}{M}\sum_{i=1}^Mx_i-\frac{1}{N}\sum_{i=1}^Nx_i\nonumber\\
	&=\frac{1}{M}\sum_{i=1}^Mx_i-(\frac{1}{N}\sum_{i=1}^{M}x_i+\frac{1}{N}\sum_{i=M+1}^Nx_i)\\
	&=(\frac{1}{M}-\frac{1}{N})\sum_{i=1}^Mx_i-\frac{1}{N}\sum_{i=M+1}^Nx_i\\
	&=\frac{1}{N}(\frac{N-M}{M}\sum_{i=1}^Mx_i-\sum_{i=M+1}^Nx_i)\\
	&=\frac{N-M}{N}(\frac{1}{M}\sum_{i=1}^Mx_i-\frac{1}{N-M}\sum_{i=M+1}^Nx_i)\\
	&\geq \frac{N-M}{N}(\frac{1}{M}\sum_{i=1}^Mx_M-\frac{1}{N-M}\sum_{i=M+1}^Nx_{M+1}) \quad \text{(because $x_i\geq x_{i+1}$)}\\
	&= \frac{N-M}{N}(x_M-x_{M+1})\\
	&\geq 0,
\end{align}
thus, we obtain $\frac{1}{M}\sum_{i=1}^Mx_i-\frac{1}{N}\sum_{i=1}^Nx_i\geq 0$, that is, $\frac{1}{M}\sum_{i=1}^Mx_i\geq \frac{1}{N}\sum_{i=1}^Nx_i$.
\end{proof}


\setcounter{proposition}{\arabic{proposition}-1}
\renewcommand{\theproposition}{1}
\renewcommand{\theHproposition}{old1}
\begin{proposition}\label{pro:truncation}
If $U_1$ can generate any $N$-qubit pure state that is a superposition of at most $K$ computational basis states, then for any $\ket{\phi}$, there exists an {\ansatz} output state $\ket{\psi}$ with $F(\ket{\phi},\ket{\psi})\geq \min\left\{\frac{K}{r}, 1\right\}$,
where $F(\ket{\phi},\ket{\psi})$ is the fidelity between $\ket{\phi}$ and $\ket{\psi}$, and $r$ is the Schmidt rank of $\ket\phi$.
\end{proposition}

\begin{proof}
For any $2N$-qubit target state $\ket{\phi}=\sum_{k=0}^{r-1}\lambda_k\ket{v_k}_A\ket{v_k}_B$, we explicitly construct an {\ansatz}, such that $\ket{\psi}=S(\hat{\bm{\theta}})\ket{0}^{\otimes 2N}$ and $F(\ket{\phi},\ket{\psi})\geq\min\{\frac{K}{r}, 1\}$, where $\hat{\bm{\theta}}=\{\hat{\bm{\theta}}_1,\hat{\bm{\theta}}_2,\hat{\bm{\theta}}_3\}$ is a certain set of parameters.

Suppose $\{\lambda_k\}_{k=0}^{r-1}$ is the Schmidt coefficients of $\ket\phi$ sorted in descending order. Similar to the proof of Proposition~\ref{pro:state}, we can choose a certain set of parameters $\hat{\bm{\theta}}$ such that
\begin{align}
	U_1(\hat{\bm{\theta}}_1)\ket{0}_A^{\otimes N} &=\frac{1}{\sqrt{M}}\sum_{k=0}^{\min\{K,r\}-1} \lambda_k\ket{k},\\
	U_2(\hat{\bm{\theta}}_2)\ket{k}_A &= \ket{v_k}_A,\\
	U_3(\hat{\bm{\theta}}_3)\ket{k}_B &=\ket{v_k}_B,
\end{align}
where $M=\sum_{k=0}^{\min\{K,r\}-1} \lambda_k^2$. Then
\begin{align}
	\ket{\psi}=S(\hat{\bm{\theta}})\ket{0}^{\otimes 2N}=\frac{1}{\sqrt{M}}\sum_{k=0}^{\min\{K,r\}-1}\lambda_k\ket{v_k}_A\ket{v_k}_B.
\end{align}
Hence
\begin{align}
	\sfid
	&=|\langle\phi|\psi\rangle|^2\\
	&=\Big(\frac{1}{\sqrt{M}}\sum_{k=0}^{\min\{K,r\}-1}\lambda_k^2\Big)^2 \\
	&=\frac{1}{M}\Big(\sum_{k=0}^{\min\{K,r\}-1}\lambda_k^2\Big)^2 \\
	&=\sum_{k=0}^{\min\{K,r\}-1}\lambda_k^2 \\
	&\geq \frac{\min\{K,r\}}{r}\sum_{k=0}^{r-1}\lambda_k^2 \quad(\text{due to lemma \ref{lemma:average}})\\
	&=\frac{\min\{K,r\}}{r}\\
	&=\min\{\frac{K}{r},1\},
\end{align}
with equality holds if and only if $\lambda_i=\lambda_j, \forall i,j=0,\dots,r-1$.
\end{proof}

\section{Proof of Proposition 2}
\label{appendix:proof-TA-NOT-2-design}
Here we prove the Proposition~\ref{pro:2-design} about SEA does not form a $2$-design. We first present Lemma~\ref{lem:cnot}, which is used in the proof.

\renewcommand{\theproposition}{S\arabic{proposition}}
\renewcommand{\theHproposition}{\theproposition}
\begin{lemma}\label{lem:cnot}
$\BCNOT=\sum_{i=0}^{2^n-1}\ketbra{i}{i}_A\otimes V_{B_i}$, where $A,B$ are subsystems and $V_{B_i}=\bigotimes_{j=0}^{n-1}X_j^{\bin{i}{j}}, \bin{i}{j}\in\{0, 1\}$, that is the operator on the subsystem $B$ that represents the binary bit of $i$ is 1 acting on pauli $X$, and $0$ acting on $I$. 
Then $V_{B_i}$ has the following properties:
\begin{enumerate}
	\item $[V_{B_i},V_{B_j}]=0$;
	\item $V_{B_i}=V_{B_i}^{\dagger}$;
	\item $\bra{0}^{\otimes n}V_{B_i}V_{B_j}\ket{0}^{\otimes n}=\delta_{ij}$,
\end{enumerate}
\end{lemma}
where $\delta_{ij}$ is Kronecker delta. 

Lemma~\ref{lem:cnot} can be easily proved using properties of Pauli operators.
\renewcommand\theproposition{2}
\renewcommand{\theHproposition}{old2}
\setcounter{proposition}{\arabic{proposition}-1}

\begin{proposition}\label{pro:2-design}
{\ansatz} with $U_i(i=1,2,3)$ being local 2-design and $\BCNOT$ as entangling layer does not form a 2-design on the global system.
\end{proposition}

\begin{proof}
Without losing generality, we only consider the SEA with $2n$ qubits in this proof. Lemma~\ref{Appendix-lem:2-design-U} says that if an ansatz forms a 2-design, the Eq.~\eqref{eq:Appendix-lem:2-design-U} should hold for any linear operators $C,D,\rho \in L(\mathbb{C}^{d^2})$. 
Assuming $d=2^{n}, C_0=D_0=\rho=\ketbra{0}{0}^{\otimes 2n}$, then if {\ansatz} forms a 2-design, we have
\begin{align}
	&\int_{U(\theta)}\bra{0}^{\otimes 2n}U^{\dagger}C_0U\rho U^{\dagger}D_0U\ket{0}^{\otimes 2n}d\eta(U) \label{eq: TA_not_2-design eq 1}\nonumber\\
	=&\int_{\mathcal{U}(d^2)}\bra{0}^{\otimes 2n}U^{\dagger}C_0U\rho U^{\dagger}D_0U\ket{0}^{\otimes 2n}d\eta(U)\\
	=&\frac{2}{d^2(d^2+1)},
	\label{eq:2n-qubit}
\end{align}
where $U(\theta)$ denotes the subset of unitary group generated by \ansatz,  and $d\eta(U)$ denotes the Haar measure.

In the following steps, we will explicitly calculate the Eq.~\eqref{eq: TA_not_2-design eq 1} and prove that it cannot be $\frac{2}{d^2(d^2+1)}$.

For the {\ansatz} with $U_i(i=1,2,3)$ being local 2-design and $\BCNOT=V=\sum_{i=0}^{2^n-1}\ketbra{i}{i}_A\otimes V_{B_i}$, the unitary of this ansatz is 
\begin{align}
	U&= \sum_{i=0}^{2^n-1}(U_2\otimes U_3)\cdot \ketbra{i}{i}\otimes V_{B_i} \cdot (U_1\otimes I_B)\label{eq:2n-U}.
\end{align}
Therefore,
\begin{align}
	U\ket{0}^{\otimes 2n}&= \sum_{i=0}^{2^n-1}(U_2\otimes U_3)\cdot \ketbra{i}{i}\otimes V_{B_i} 
	\cdot (U_1\otimes I_B)\ket{0}^{\otimes 2n},\\
	\bra{0}^{\otimes 2n}U^{\dagger} &=\sum_{i=0}^{2^n-1}\bra{0}^{\otimes 2n}(U_1^{\dagger}\otimes I_B) \cdot\ketbra{i}{i}\otimes V_{B_i}
	\cdot (U_2^{\dagger}\otimes U_3^{\dagger}).
\end{align}

Then
\begin{align}
	&\quad\bra{0}^{\otimes 2n}U^{\dagger}C_0U\rho U^{\dagger}D_0U\ket{0}^{\otimes 2n}\nonumber\\
	&=\bra{0}^{\otimes 2n}U^{\dagger}(\ketbra{0}{0})^{\otimes 2n}U(\ketbra{0}{0})^{\otimes 2n}U^{\dagger}
	(\ketbra{0}{0})^{\otimes 2n}U\ket{0}^{\otimes 2n}\\
	&=\bra{0}^{\otimes 2n}(U_1^{\dagger}\otimes I_B) \cdot V^{\dagger}\cdot (U_2^{\dagger}\otimes U_3^{\dagger})(\ketbra{0}{0})^{\otimes 2n}
	(U_2\otimes U_3)\cdot V\cdot (U_1\otimes I_B)(\ketbra{0}{0})^{\otimes 2n}(U_1^{\dagger}\otimes I_B) 
	\cdot V^{\dagger}\nonumber\\
	&\quad \cdot (U_2^{\dagger}\otimes U_3^{\dagger}))(\ketbra{0}{0})^{\otimes 2n}(U_2\otimes U_3)\cdot V
	\cdot (U_1\otimes I_B)\ket{0}^{\otimes 2n}\\
	&:=\bra{0}^{\otimes 2n}(U_1^{\dagger}\otimes I_B) \cdot V^{\dagger}\cdot (U_2^{\dagger}\otimes U_3^{\dagger})\cdot C_0
	\cdot(U_2\otimes U_3)\cdot\rho_1\cdot(U_2^{\dagger}\otimes U_3^{\dagger})\cdot D_0\cdot(U_2\otimes U_3)
	\cdot V\nonumber\\
	&\quad\cdot (U_1\otimes I_B)\ket{0}^{\otimes 2n},
\end{align}
where 
\begin{align}
	V&=\sum_{i=0}^{2^n-1}\ketbra{i}{i}_A\otimes V_{B_i},\\
	\rho_1&=V\cdot (U_1\otimes I_B)(\ketbra{0}{0})^{\otimes 2n}(U_1^{\dagger}\otimes I_B) \cdot V^{\dagger}.
\end{align}
Therefore, we have
\begin{align}
	\rho_{1}^A & = \tr_B[\rho_1]\\
	&=\sum_{k=0}^{d-1}I_A\otimes\bra{k}_B\cdot V\cdot (U_1\otimes I_B)(\ketbra{0}{0})^{\otimes 2n}\cdot (U_1^{\dagger}\otimes I_B) \cdot V^{\dagger}\cdot I_A\otimes\ket{k}_B\\
	&=\sum_{k=0}^{d-1}I_A\otimes\bra{k}_B\cdot (\sum_{i=0}^{d-1}\ketbra{i}{i}_A\otimes V_{B_i})\cdot (U_1\otimes I_B)(\ketbra{0}{0})^{\otimes 2n}\cdot (U_1^{\dagger}\otimes I_B) \cdot (\sum_{j=0}^{d-1}\ketbra{j}{j}_A\otimes V_{B_j}^{\dagger})\cdot I_A\otimes\ket{k}_B\\
	&=\sum_{i,j,k=0}^{d-1}I_A\otimes\bra{k}_B\cdot (\ketbra{i}{i}_A\otimes V_{B_i})\cdot (U_1\otimes I_B)(\ketbra{0}{0})^{\otimes 2n}\cdot (U_1^{\dagger}\otimes I_B) \cdot (\ketbra{j}{j}_A\otimes V_{B_j}^{\dagger})\cdot I_A\otimes\ket{k}_B\\
	&=\sum_{i,j,k=0}^{d-1}\ketbra{i}{i}_AU_1(\ketbra{0}{0})^{\otimes n}U_1^{\dagger}\ketbra{j}{j}_A\otimes\bra{k}_BV_{B_i}(\ketbra{0}{0})^{\otimes n}V_{B_j}^{\dagger}\ket{k}_B\\
	&=\sum_{i,j=0}^{d-1}\ketbra{i}{i}_AU_1(\ketbra{0}{0})^{\otimes n}U_1^{\dagger}\ketbra{j}{j}_A(\sum_{k=0}^{d-1}\bra{k}_BV_{B_i}(\ketbra{0}{0})^{\otimes n}V_{B_j}^{\dagger}\ket{k}_B)\\
	&=\sum_{i,j=0}^{d-1}\ketbra{i}{i}_AU_1(\ketbra{0}{0})^{\otimes n}U_1^{\dagger}\ketbra{j}{j}_A\tr(V_{B_i}(\ketbra{0}{0})^{\otimes n}V_{B_j}^{\dagger})\\
	&=\sum_{i,j=0}^{d-1}\ketbra{i}{i}_AU_1(\ketbra{0}{0})^{\otimes n}U_1^{\dagger}\ketbra{j}{j}_A\tr(\bra{0}^{\otimes n}V_{B_i}V_{B_j}^{\dagger}\ket{0}^{\otimes n})\\
	&=\sum_{i=0}^{d-1}\ketbra{i}{i}U_1\ketbra{0}{0}^{\otimes n}U_1^{\dagger}\ketbra{i}{i},\quad \text{(due to lemma~\ref{lem:cnot})}\\
	\rho_{1}^B& =  \tr_A[\rho_1]\\
	&=\sum_{k=0}^{d-1}\bra{k}_A\otimes I_B\cdot V\cdot (U_1\otimes I_B)(\ketbra{0}{0})^{\otimes 2n}\cdot (U_1^{\dagger}\otimes I_B) \cdot V^{\dagger}\cdot \ket{k}_A\otimes I_B\\
	&=\sum_{k=0}^{d-1}\bra{k}_A\otimes I_B\cdot  (\sum_{i=0}^{d-1}\ketbra{i}{i}_A\otimes V_{B_i})\cdot (U_1\otimes I_B)(\ketbra{0}{0})^{\otimes 2n}\cdot (U_1^{\dagger}\otimes I_B) \cdot (\sum_{j=0}^{d-1}\ketbra{j}{j}_A\otimes V_{B_j}^{\dagger})\cdot \ket{k}_A\otimes I_B\\  
	&=\sum_{i,j,k=0}^{d-1}\bra{k}_A\ketbra{i}{i}_AU_1(\ketbra{0}{0})^{\otimes n}U_1^{\dagger}\ketbra{j}{j}_A\ket{k}_A\cdot V_{B_i}(\ketbra{0}{0})^{\otimes n}V_{B_j}\\
	&=\sum_{k=0}^{d-1}\bra{k}U_1\ketbra{0}{0}^{\otimes n}U_1^{\dagger}\ket{k}\cdot V_{B_k}\ketbra{0}{0}^{\otimes n}V_{B_k}.    
\end{align}

Then combining lemma~\ref{Appendix-lem:2-design-U} and lemma~\ref{lem:2-design-UxU}, we have
\begin{align}
	& \quad\int_{U(\theta)}\bra{0}^{\otimes 2n}U^{\dagger}C_0U\rho U^{\dagger}D_0U\ket{0}^{\otimes 2n}d\eta(U)\nonumber\\
	&=\int_{U(\theta)}\bra{0}^{\otimes 2n}(U_1^{\dagger}\otimes I_B) \cdot V^{\dagger}\cdot (U_2^{\dagger}\otimes U_3^{\dagger})\cdot C_0
	\cdot(U_2\otimes U_3)\cdot\rho_1\cdot(U_2^{\dagger}\otimes U_3^{\dagger})\cdot D_0\cdot(U_2\otimes U_3)
	\cdot V \nonumber\\
	&\quad\cdot (U_1\otimes I_B)\ket{0}^{\otimes 2n}\\
	&=\int_{\mathcal{U}_1(d)}\bra{0}^{\otimes 2n}(U_1^{\dagger}\otimes I_B) \cdot V^{\dagger}
	\cdot(t_0\rho_1+t_1\frac{\rho_{1}^A\otimes I_B}{d}+t_2\frac{I_A\otimes \rho_{1}^B}{d}
	+t_3 I_{AB}\tr(\rho_1))\cdot V\nonumber\\
	&\quad \cdot (U_1\otimes I_B)\ket{0}^{\otimes 2n}d\eta(U)\\
	&:=M_0+M_1+M_2+M_3\label{eq:U1U2},
\end{align}
where
\begin{align}
	M_0&=t_0\int_{\mathcal{U}_1(d)}\bra{0}^{\otimes 2n}(U_1^{\dagger}\otimes I_B) \cdot V^{\dagger}
	\cdot\rho_1\cdot V\cdot (U_1\otimes I_B)\ket{0}^{\otimes 2n}d\eta(U)\\
	& =t_0,\\
	M_1&=\frac{t_1}{d}\int_{\mathcal{U}_1(d)}\bra{0}^{\otimes 2n}(U_1^{\dagger}\otimes I_B) \cdot V^{\dagger}\cdot\rho_{1}^A\otimes I_B
	\cdot V\cdot (U_1\otimes I_B)\ket{0}^{\otimes 2n}d\eta(U),\\
	M_2&=\frac{t_2}{d}\int_{\mathcal{U}_1(d)}\bra{0}^{\otimes 2n}(U_1^{\dagger}\otimes I_B) \cdot V^{\dagger}\cdot I_A\otimes \rho_{1}^B
	\cdot V\cdot (U_1\otimes I_B)\ket{0}^{\otimes 2n}d\eta(U),\\
	M_3&=t_3\int_{\mathcal{U}_1(d)}\bra{0}^{\otimes 2n}(U_1^{\dagger}\otimes I_B) \cdot V^{\dagger}\cdot I_{AB}\tr(\rho_1)
	\cdot V\cdot (U_1\otimes I_B)\ket{0}^{\otimes 2n}d\eta(U)\\
	&=t_3,\\
	t_0&=t_3=\frac{1}{d^2(d+1)^2}, \\
	t_1&=t_2=\frac{1}{d(d+1)^2}.
\end{align}

Now the problem is how do we calculate $M_1$ and $M_2$. Combine with $V=\sum_{i=0}^{d-1}\ketbra{i}{i}_A\otimes V_{B_i}$, there has
\begin{align}
	M_1&=\frac{t_1}{d}\int_{\mathcal{U}_1(d)}\bra{0}^{\otimes 2n}(U_1^{\dagger}\otimes I_B) \cdot V^{\dagger}\cdot\rho_{1}^A\otimes I_B\cdot V\cdot (U_1\otimes I_B)\ket{0}^{\otimes 2n}d\eta(U)\\
	&=\frac{t_1}{d}\int_{\mathcal{U}_1(d)}\bra{0}^{\otimes 2n}(U_1^{\dagger}\otimes I_B)(\sum_{i=0}^{d-1}\ketbra{i}{i}\otimes V_{B_i}^{\dagger})
	\cdot(\sum_{k=0}^{d-1}\ketbra{k}{k}U_1\ketbra{0}{0}^{\otimes n}U_1^{\dagger}\ketbra{k}{k}\otimes I_{B})
	\nonumber\\
	&\quad
	\cdot (\sum_{j=0}^{d-1}\ketbra{j}{j}\otimes V_{B_j})(U_1\otimes I_B)\ket{0}^{\otimes 2n}d\eta(U)\\
	&=\frac{t_1}{d}\sum_{i=0}^{d-1}\int_{\mathcal{U}_1(d)}\bra{0}^{\otimes n}U_1^{\dagger}\ketbra{i}{i}U_1(\ketbra{0}{0})^{\otimes n}U_1^{\dagger}
	\ketbra{i}{i}U_1\ket{0}^{\otimes n}\cdot (\braket{0}{0})^{\otimes n}d\eta(U) \\
	&=\frac{2t_1}{d(d+1)},\label{eq:M1}\\
	M_2&=\frac{t_2}{d}\int_{\mathcal{U}_1(d)}\bra{0}^{\otimes 2n}(U_1^{\dagger}\otimes I_B) \cdot V^{\dagger}\cdot I_A\otimes\rho_{1}^B\cdot V\cdot (U_1\otimes I_B)\ket{0}^{\otimes 2n}d\eta(U)\\
	&=\frac{t_2}{d}\int_{\mathcal{U}_1(d)}\bra{0}^{\otimes 2n}(U_1^{\dagger}\otimes I_B)(\sum_{i=0}^{d-1}\ketbra{i}{i}\otimes V_{B_i}^{\dagger})
	\cdot(I_A\otimes\sum_{k=0}^{d-1}\bra{k}U_1\ketbra{0}{0}^{\otimes n}U_1^{\dagger}\ket{k} V_{B_k}\ketbra{0}{0}^{\otimes n}V_{B_k})
	\nonumber\\
	&\quad \cdot (\sum_{j=0}^{d-1}\ketbra{j}{j}\otimes V_{B_j})(U_1\otimes I_B)\ket{0}^{\otimes 2n}d\eta(U)\\
	&=\frac{t_2}{d}\sum_{i=0}^{d-1}\int_{\mathcal{U}_1(d)}\bra{0}^{\otimes n}U_1^{\dagger}\ketbra{i}{i}U_1(\ketbra{0}{0})^{\otimes n}U_1^{\dagger} \ketbra{i}{i}U_1\ket{0}^{\otimes n}\cdot (\braket{0}{0})^{\otimes n}d\eta(U)\\
	&=\frac{2t_2}{d(d+1)}.\label{eq:M2}
\end{align}
Therefore, 
\begin{align}
	{\rm Eq.}~\eqref{eq:U1U2}&=M_0+M_1+M_2+M_3\\
	&=\frac{2d+6}{d^2(d+1)^3}\\
	&\neq\frac{2}{d^2(d^2+1)}(\text{$n\geq1,d=2^n>1$})\\
	&={\rm Eq.}~\eqref{eq:2n-qubit}.
\end{align}
Thus, {\ansatz} with $U_i(i=1,2,3)$ being local 2-design ansatzes and $\BCNOT$ is not a unitary 2-design.
\end{proof}





\section{Proof of Proposition 3}
\label{appendix:trainability-analysis-sea}

In this section, we present a proof of Proposition~\ref{pro:var} by calculating the variance of the cost gradient of the SEA with local $2$-designs.
\renewcommand\theproposition{3}
\renewcommand{\theHproposition}{old3}
\setcounter{proposition}{\arabic{proposition}-1}
\begin{proposition}\label{pro:var}
For an SEA defined on $2N$ qubits with all sub-blocks being local $2$-designs, the variance of the cost gradient scales with the number of qubits as
\begin{equation}
	\var_{\rm SEA}[\partial_\mu C]\in\mathcal{O}(2^{-(N+M)}),
\end{equation}
where $M\in\{0,1,...,N\}$ denotes the number of $\CNOT$ gates used in the entangling layer and the variance is taken over all SEA sub-blocks independently.
\end{proposition}
\begin{proof}
Suppose $\mathbf{U} = ( U_2 \otimes U_3)\BCNOT(U_1\otimes I_B)$ is an SEA on a equally bipartite system $AB$ of dimension $d=d_Ad_B$ and $d_A=d_B=2^{N}$. Note that the $\BCNOT$ in $\mathbf{U}$ represent a collection of $M$ $\CNOT$ gates between subsystem $A$ and $B$. We focus on a single parameter $\theta$ within $U_1$ such that $U_1 = U_5 \left(e^{-i\Omega\theta}\right) U_4$. The case where $\theta$ locates at $U_2$ or $U_3$ can be analysed similarly. Denote $W=(U_2\otimes U_3)\BCNOT(U_5\otimes I_B)$, $V=U_4\otimes I_B$ and then we have $\mathbf{U}=W\left(e^{-i\Omega\theta}\otimes I_B\right)V$. The cost gradient becomes
\begin{equation}
	\begin{aligned}
		\frac{\partial C}{\partial \theta} 
		& = i \tr\left( V \rho V^\dagger \left[ \Omega\otimes I_B, \hat{W}^\dagger H \hat{W} \right]  \right),
	\end{aligned}
\end{equation}
where $\hat{W} = W \left( e^{-i\Omega\theta}\otimes I_B \right)$ and $\hat{U}_5 = U_5  \left( e^{-i\Omega\theta} \right)$. Suppose all $U_i$ are sampled from local $2$-designs. Since the tensor product of two $1$-designs is still a $1$-design, the expectation of the cost gradient is zero. Then consider the variance
\begin{equation}
	\var_{\rm SEA} \left[ \partial_\mu C \right] = - \mathbb{E}\left[ \left( \tr\left( V \rho V^\dagger \left[ \Omega\otimes I_B, \hat{W}^\dagger H \hat{W} \right] \right) \right)^2 \right],
\end{equation}
where the variance is integrated with respect to $U_2,U_3,U_4$ and $U_5$. We exploit the RTNI package~\cite{Fukuda2019} to calculate this integral. It turns out that the exact expression of the variance is dominant by
\begin{equation}\label{Eq:var_sea}
	\var_{\rm SEA} \left[ \partial_\mu C \right] \xrightarrow{d\rightarrow\infty} \frac{2 \tr\left( (\tr_B H)^2 \right) \tr\left(\Omega^2\right) }{d_A (d_A^2-1)^2 (d_B^2-1)} \cdot \left(
	\diagram{
		\def\r{0.5};
		\def\spacex{1.4};
		\def\spacey{1.4};
		\def\xc{0};
		\def\xct{\xc+2*\r};
		\def\xrt{\xc+5*\r};
		\def\xlt{\xc-\r};
		\def\yc{0};
		\def\yu{1*\spacey};
		\def\yuu{2*\spacey};
		\def\yuuu{3*\spacey};
		\coordinate (pc) at (\xct,\yc);
		\coordinate (pu) at (\xct,\yu);
		\coordinate (puu) at (\xct,\yuu);
		\coordinate (puuu) at (\xct,\yuuu);
		
		\coordinate (pxryuu) at (\xrt,\yuu);
		\coordinate (pxryu) at (\xrt,\yu);
		
		\coordinate (pxlyuu) at (\xlt, \yuu);
		\coordinate (pxlyu) at (\xlt, \yu);
		
		\draw (\xlt-\r,\yuuu-0.5*\r) .. controls (\xlt-1.5*\r,\yuuu-0.5*\r) and (\xlt-1.5*\r,\yuu+0.5*\r) .. (\xlt-\r,\yuu+0.5*\r);
		\draw (\xlt-\r,\yu-0.5*\r) .. controls (\xlt-1.5*\r,\yu-0.5*\r) and (\xlt-1.5*\r,\yc+0.5*\r) .. (\xlt-\r,\yc+0.5*\r);
		
		\draw (\xrt+\r,\yuuu-0.5*\r) .. controls (\xrt+1.5*\r,\yuuu-0.5*\r) and (\xrt+1.5*\r,\yuu+0.5*\r) .. (\xrt+\r,\yuu+0.5*\r);
		\draw (\xrt+\r,\yu-0.5*\r) .. controls (\xrt+1.5*\r,\yu-0.5*\r) and (\xrt+1.5*\r,\yc+0.5*\r) .. (\xrt+\r,\yc+0.5*\r);
		
		\draw (\xlt-\r,\yuu-0.5*\r) .. controls (\xlt-1.5*\r,\yuu-0.5*\r) and (\xlt-1.5*\r,\yu+0.5*\r) .. (\xlt-\r,\yu+0.5*\r);
		\draw (\xct-\r,\yuu-0.5*\r) .. controls (\xct-1.5*\r,\yuu-0.5*\r) and (\xct-1.5*\r,\yu+0.5*\r) .. (\xct-\r,\yu+0.5*\r);
		\draw (\xrt-\r,\yuu-0.5*\r) .. controls (\xrt-1.5*\r,\yuu-0.5*\r) and (\xrt-1.5*\r,\yu+0.5*\r) .. (\xrt-\r,\yu+0.5*\r);
		
		\draw (\xrt+\r,\yu+0.5*\r) .. controls (\xrt+1.5*\r,\yu+0.5*\r) and (\xrt+1.5*\r,\yuu-0.5*\r) .. (\xrt+\r,\yuu-0.5*\r);
		\draw (\xct+\r,\yu+0.5*\r) .. controls (\xct+1.5*\r,\yu+0.5*\r) and (\xct+1.5*\r,\yuu-0.5*\r) .. (\xct+\r,\yuu-0.5*\r);
		\draw (\xlt+\r,\yu+0.5*\r) .. controls (\xlt+1.5*\r,\yu+0.5*\r) and (\xlt+1.5*\r,\yuu-0.5*\r) .. (\xlt+\r,\yuu-0.5*\r);

		\draw (\xc-2*\r,\yuuu-0.5*\r) -- (\xrt+\r,\yuuu-0.5*\r);
		\draw (\xc-2*\r,\yuu+0.5*\r) -- (\xrt+\r,\yuu+0.5*\r);
		\draw (\xc-2*\r,\yu-0.5*\r) -- (\xrt+\r,\yu-0.5*\r);
		\draw (\xc-2*\r,\yc+0.5*\r) -- (\xrt+\r,\yc+0.5*\r);
		
		\draw[ten, shift=(puu)] (-\r,-\r) rectangle (\r,\r);
		\node at (puu) {\scriptsize $\rho^t$};
		\draw[ten, shift=(pu)] (-\r,-\r) rectangle (\r,\r);
		\node at (pu) {\scriptsize $\rho^*$};
		
		\draw[ten, shift=(pxryuu)] (-\r,-\r) rectangle (\r,\r);
		\node at (pxryuu) {\scriptsize $C^t$};
		\draw[ten, shift=(pxryu)] (-\r,-\r) rectangle (\r,\r);
		\node at (pxryu) {\scriptsize $C$};
		
		\draw[ten, shift=(pxlyuu)] (-\r,-\r) rectangle (\r,\r);
		\node at (pxlyuu) {\scriptsize $C^t$};
		\draw[ten, shift=(pxlyu)] (-\r,-\r) rectangle (\r,\r);
		\node at (pxlyu) {\scriptsize $C$};
	}
	\right),
\end{equation}
where the last factor is represented by the tensor network notation for convenience and $C$ is $\BCNOT$ for short. The superscript ``$*$'' denotes complex conjugate, and the superscript ``$t$'' denotes the rearrangement of the indices of $\mathcal{H}_A$, $\mathcal{H}_B$, i.e.,
\begin{equation}
	\left(
	\diagram{
		\def\r{0.9};
		\def\spacex{1.4};
		\def\spacey{1.4};
		\def\xc{0};
		\def\yc{0};
		\def\yu{0.5*\r};
		\def\yd{-0.5*\r};
		\coordinate (pc) at (\xc,\yc);
		\draw (\xc-2*\r,\yu) -- (\xc-\r,\yu);
		\draw (\xc+\r,\yu) -- (\xc+2*\r,\yu);
		\draw (\xc-2*\r,\yd) -- (\xc-\r,\yd);
		\draw (\xc+\r,\yd) -- (\xc+2*\r,\yd);
		\draw[ten, shift=(pc)] (-\r,-\r) rectangle (\r,\r);
		\node at (pc) {\normalsize $C$};
		\node at (\xc-3*\r,\yu) {\footnotesize $\mathcal{H}_A$};
		\node at (\xc+3*\r,\yu) {\footnotesize $\mathcal{H}_A$};
		\node at (\xc-3*\r,\yd) {\footnotesize $\mathcal{H}_B$};
		\node at (\xc+3*\r,\yd) {\footnotesize $\mathcal{H}_B$};
	}\right),~~\left(
	\diagram{
		\def\r{0.9};
		\def\spacex{1.4};
		\def\spacey{1.4};
		\def\xc{0};
		\def\yc{0};
		\def\yu{0.5*\r};
		\def\yd{-0.5*\r};
		\coordinate (pc) at (\xc,\yc);
		\draw (\xc-2*\r,\yu) -- (\xc-\r,\yu);
		\draw (\xc+\r,\yu) -- (\xc+2*\r,\yu);
		\draw (\xc-2*\r,\yd) -- (\xc-\r,\yd);
		\draw (\xc+\r,\yd) -- (\xc+2*\r,\yd);
		\draw[ten, shift=(pc)] (-\r,-\r) rectangle (\r,\r);
		\node at (pc) {\normalsize $C^t$};
		\node at (\xc-3*\r,\yu) {\footnotesize $\mathcal{H}_B$};
		\node at (\xc+3*\r,\yu) {\footnotesize $\mathcal{H}_B$};
		\node at (\xc-3*\r,\yd) {\footnotesize $\mathcal{H}_A$};
		\node at (\xc+3*\r,\yd) {\footnotesize $\mathcal{H}_A$};
	}\right).
\end{equation}
Note that a $2$-qubit $\CNOT$ gate contribute one loop in (\ref{Eq:var_sea}) and a $2$-qubit identity contributes two loops in (\ref{Eq:var_sea}). Each loop contributes a factor $2$. Considering the contributions from $\tr\left( (\tr_B H)^2 \right)\sim d_Ad_B^2$ and $\tr(\Omega^2)\sim d_A$, the dominant term shall scale with the number of qubits $2N$ as $\mathcal{O}(2^{-(N+M)})$. Note that here we suppose the Hamiltonian $H$ contains terms acting on $B$ trivially, which is reasonable for physical and chemical models with local interactions. Compared to $\mathcal{O}(2^{-2N})$ of the hardware-efficient ansatz forming a global $2$-design~\cite{mcclean2018barren,cerezo2021cost}, we find that the SEA could provide a square root advantage on the gradient magnitude with the number of qubits over the hardware-efficient ansatz with sufficient depths, while keeping the accuracy of expressing the low-entangled target states.
\end{proof}

\section{Proof of Proposition 4}\label{appendix:expressibility-analysis-sea}
In this section we provide a proof of Proposition~\ref{pro:frame potential} to estimate the expressibility of SEA, which is quantified by the so-called $t$-degree frame potential~\cite{sim2019expressibility}, i.e.,
\begin{equation}
\mathcal{F}^{(t)} = \mathbb{E}\left[ \left( \tr \left( \mathbf{U} \rho \mathbf{U}^\dagger \mathbf{V} \rho \mathbf{V}^\dagger \right) \right)^{t} \right].
\end{equation}
Here $\rho=\proj{0}$ is the zero state on a equally bipartite system $AB$ of dimension $d=d_Ad_B$ and $d_A=d_B=2^{N}$. The expectation is taken over two copies of the ansatz ensemble $\mathbf{U},\mathbf{V}\in\mathbb{U}$ independently. For comparison, we first show the results of the Haar measure. For $t=1$ and $t=2$, we have
\begin{align}
\mathcal{F}^{(1)}_{\text{Haar}} & = \mathbb{E}\left[ \tr \left( \mathbf{U} \rho \mathbf{U}^\dagger \mathbf{V} \rho \mathbf{V}^\dagger \right) \right] = \frac{\tr(\rho)}{d} \mathbb{E}\left[ \tr \left( \mathbf{V} \rho \mathbf{V}^\dagger \right) \right] = \frac{ \left( \tr \rho \right)^2 }{d} = \frac{1}{d},\\
\mathcal{F}^{(2)}_{\text{Haar}} & = \mathbb{E} \left[ \left( \tr \left( \mathbf{U} \rho \mathbf{U}^\dagger \mathbf{V} \rho \mathbf{V}^\dagger \right) \right)^2 \right] \\
& = \frac{1}{d^2-1} \left[ (\tr\rho)^2 \left((\tr \rho)^2 - \frac{\tr(\rho^2)}{d} \right) + \tr(\rho^2) \left(\tr(\rho^2) - \frac{ (\tr \rho)^2 }{d} \right) \right] \\
& = \frac{1}{d^2-1} \left( (\tr\rho)^4 + \tr(\rho^2)^2 - \frac{2(\tr\rho)^2 \tr(\rho^2)}{d} \right) = \frac{2}{d(d+1)}.
\end{align}
We can see that the second frame potential is exponentially small in the number of qubits $2N$ with scaling $O(2^{-4N})$. 

\renewcommand\theproposition{4}
\renewcommand{\theHproposition}{old4}
\setcounter{proposition}{\arabic{proposition}-1}
\begin{proposition}\label{pro:frame potential}
For an SEA defined on $2N$ qubits with all sub-blocks being local $2$-designs, the first and second frame potential satisfies
\begin{align}
	&\mathcal{F}^{(1)}_{\rm SEA} = 2^{-2N},\\
	&\mathcal{F}^{(2)}_{\rm SEA} \in \mathcal{O}(2^{-4N}).
\end{align}
\end{proposition}
\begin{proof}
For SEAs with sub-blocks being local $2$-designs, the $1$-degree frame potential is

\begin{align}
	\mathcal{F}^{(1)}_{\rm SEA} = &~\mathbb{E} \left[ \tr \left( \mathbf{U} \rho \mathbf{U}^\dagger \mathbf{V} \rho \mathbf{V}^\dagger \right) \right] \\
	= &~\mathbb{E} \Big[ \tr \Big( (U_2\otimes U_3) \CNOT (U_1\otimes I_B) \rho (U_1^\dagger \otimes I_B) \CNOT (U_2^\dagger \otimes U_3^\dagger) \nonumber\\
	&~(V_2\otimes V_3) \CNOT (V_1\otimes I_B) \rho (V_1^\dagger \otimes I_B) \CNOT (V_2^\dagger \otimes V_3^\dagger) \Big) \Big] \\
	= &~\mathbb{E} \Big[ \tr \Big( (U_2\otimes U_3) \CNOT (U_1\otimes I_B) \rho (U_1^\dagger \otimes I_B) \CNOT (U_2^\dagger \otimes U_3^\dagger)\nonumber \\
	&~ \CNOT (V_1\otimes I_B) \rho (V_1^\dagger \otimes I_B) \CNOT \Big) \Big], 
\end{align}
where we have used the unitary invariance of the Haar measure to eliminate $V_2$ and $V_3$. Integrating with respect to $V_1,U_1$, and then $U_2,U_3$ leads to
\begin{align}
	\mathcal{F}^{(1)}_{\rm SEA} = &~\frac{1}{d_A^2} \mathbb{E} \Big[ \tr \Big( (U_2\otimes U_3) \CNOT (I_A \otimes \tr_A(\rho) ) \CNOT (U_2^\dagger \otimes U_3^\dagger) \nonumber\\
	&~ \CNOT ( I_A\otimes \tr_A(\rho) ) \CNOT \Big) \Big] \\
	= &~ \frac{1}{d_A^3 d_B} \left( \tr \left( \CNOT (I_A \otimes \tr_A(\rho) ) \CNOT \right) \right)^2 = \frac{1}{d_A d_B} = \frac{1}{d}.
\end{align}
That is to say, the $1$-degree frame potential of SEA with local $2$-designs is optimal, which can be easily seen from the fact that SEA with local $2$-designs is an exact $1$-design. The second frame potential $\mathcal{F}^{(2)}$ can be calculated similarly. Similar with the proof in the last section, we again exploit the RTNI package~\cite{Fukuda2019} to calculate this integral. It turns out that the exact expression of $\mathcal{F}^{(2)}_{\rm SEA}$ with full $\CNOT$s is also dominant by
\begin{equation}\label{Eq:fp_sea}
	\mathcal{F}^{(2)}_{\rm SEA} \xrightarrow{d\rightarrow\infty} \frac{ 2 }{(d_A^2-1)^3 (d_B^2-1)} \cdot \left(
	\diagram{
		\def\r{0.5};
		\def\spacex{1.4};
		\def\spacey{1.4};
		\def\xc{0};
		\def\xct{\xc+2*\r};
		\def\xrt{\xc+5*\r};
		\def\xlt{\xc-\r};
		\def\yc{0};
		\def\yu{1*\spacey};
		\def\yuu{2*\spacey};
		\def\yuuu{3*\spacey};
		\coordinate (pc) at (\xct,\yc);
		\coordinate (pu) at (\xct,\yu);
		\coordinate (puu) at (\xct,\yuu);
		\coordinate (puuu) at (\xct,\yuuu);
		
		\coordinate (pxryuu) at (\xrt,\yuu);
		\coordinate (pxryu) at (\xrt,\yu);
		
		\coordinate (pxlyuu) at (\xlt, \yuu);
		\coordinate (pxlyu) at (\xlt, \yu);

		\draw (\xlt-\r,\yu-0.5*\r) .. controls (\xlt-2*\r,\yu-0.5*\r) and (\xlt-2*\r,\yuu+0.5*\r) .. (\xlt-\r,\yuu+0.5*\r);
		
		\draw (\xlt-\r,\yuu-0.5*\r) .. controls (\xlt-1.5*\r,\yuu-0.5*\r) and (\xlt-1.5*\r,\yu+0.5*\r) .. (\xlt-\r,\yu+0.5*\r);
		
		\draw (\xlt+\r,\yu+0.5*\r) .. controls (\xlt+1.5*\r,\yu+0.5*\r) and (\xlt+1.5*\r,\yuu-0.5*\r) .. (\xlt+\r,\yuu-0.5*\r);

		\draw (\xc-2*\r,\yuu+0.5*\r) -- (\xc+\r,\yuu+0.5*\r);
		\draw (\xc-2*\r,\yu-0.5*\r) -- (\xc+\r,\yu-0.5*\r);
		
		\filldraw[black,fill=tensorblue] (\xc+\r,\yuu+0.5*\r-0.5*\r)--(\xc+\r+\r,\yuu+0.5*\r)--(\xc+\r,\yuu+0.5*\r+0.5*\r)--(\xc+\r,\yuu+0.5*\r-0.5*\r);
		\node at (\xc+\r+0.3*\r,\yuu+0.5*\r) {\tiny $0$};
		\filldraw[black,fill=tensorblue] (\xc+\r,\yu-0.5*\r-0.5*\r)--(\xc+\r+\r,\yu-0.5*\r)--(\xc+\r,\yu-0.5*\r+0.5*\r)--(\xc+\r,\yu-0.5*\r-0.5*\r);
		\node at (\xc+\r+0.3*\r,\yu-0.5*\r) {\tiny $0$};
		
		\draw[ten, shift=(pxlyuu)] (-\r,-\r) rectangle (\r,\r);
		\node at (pxlyuu) {\scriptsize $C^t$};
		\draw[ten, shift=(pxlyu)] (-\r,-\r) rectangle (\r,\r);
		\node at (pxlyu) {\scriptsize $C$};
	}
	\right)^4 
	= \frac{ 2 d_A^4 }{(d_A^2-1)^3 (d_B^2-1)},
\end{equation}
of scaling $\mathcal{O}(2^{-4N})$, which has no difference with the Haar case in the sense of scaling, though the exact value is larger.
\end{proof}

This can be easily understood if we replace CNOT gates with identities as the tensor product ansatzes (TEN) in ~\cite{nakaji2021expressibility}, which is easier to compute compactly. Specifically, suppose the number of qubits $2N$ is divisible by $k$ so that the $2N$ qubits could be divided equally into $k$ subsystems with $2N/k$ qubits. The tensor product ansatz is defined by $\mathbf{U}=\bigotimes_{i=1}^{k} U_i$ with each unitary $U_i$ acting on each subsystem. If the ensembles of $U_i$ are all unitary $2$-designs, then the following equality hold
\begin{equation}
\mathcal{F}^{(2)}_{\text{\rm TEN}} = 2^{k-1} \cdot \frac{2^{2N}+1}{(2^{2N/k}+1)^{k}} \cdot \mathcal{F}^{(2)}_{\text{\rm Haar}}.
\end{equation}
For large $2N$ and fixed $k$, we have $\mathcal{F}^{(2)}_{\text{\rm TEN}}\approx 2^{k-1} \mathcal{F}^{(2)}_{\text{\rm Haar}}$ which means that there is no difference in the sense of scaling between the finite tensor product ensemble and the Haar ensemble. They both behave as $\mathcal{O}(2^{-4N})$. This partially explains why the second frame potential of the SEA can have the same scaling as the extreme one.

\section{Supplementary Description of Experiments}
\label{appendix:exp}
In this section, we present results from supplementary numerical experiments. Table~\ref{tab:param} shows the specific calculation method for the numbers of parameters. Fig.~\ref{fig:VQE-LiH} displays the results of VQE using the SEA of Schmidt coefficient layer as $R_y(\bm{\theta})^{\otimes 6}$, entangling layer as 6 CNOTs and the LBC layers as two $6$-qubit ALTs. Fig.~\ref{Fig: BP-sub.2} shows the variance of the largest partial derivative in each sample.

\label{app:exp}
\renewcommand{\thetable}{S\arabic{table}}
\begin{table}[htbp]
\centering
\caption{Comparison of the number of parameters in different ansatzes. $N$ represents the half of qubit number, $l_i(i=1,2,3)$ is the number of layer of $U_i$.}
\begin{ruledtabular}
	\begin{tabular}{cccc}
		& SEA & ALT & Random \\
		\colrule
		parameter number & $3N+6(N-1)l_1$ & $2N+2(2N-1)l_2$ &  $2Nl_3$
	\end{tabular}
\end{ruledtabular}
\label{tab:param}
\end{table}

\renewcommand{\thefigure}{S\arabic{figure}}
\begin{figure}[htbp]
\centering
\includegraphics[width=0.45\textwidth]{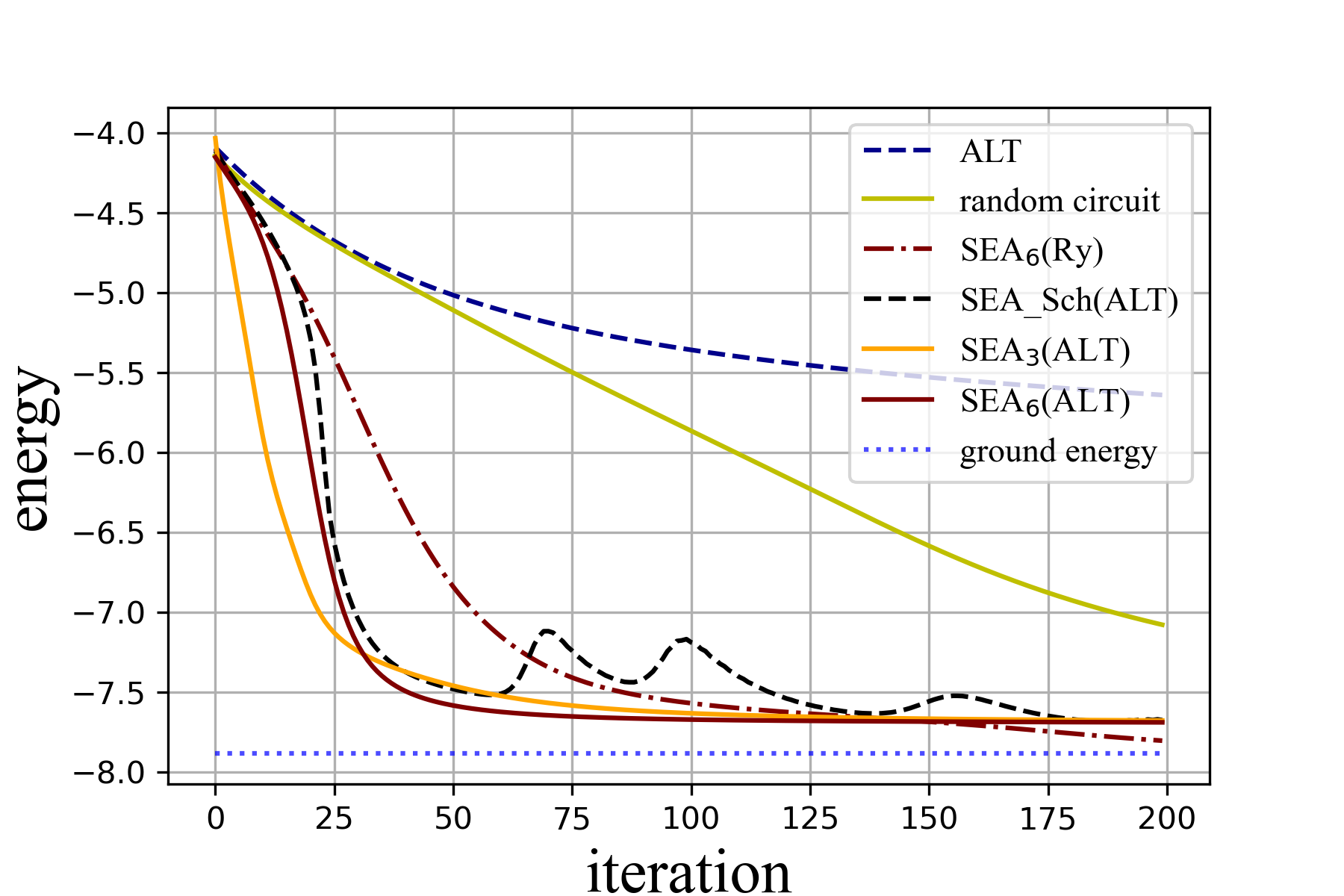}    
\caption{\textbf{Numerical experiment of VQE on LiH (12-qubit)}.The blue dotted line is the theoretical ground energy of LiH, and the lines from top to bottom represent the experimental results of ALT, the random circuit, {\ansatz} with $R_y(\bm{\theta})$ as $U_1$ and two ALTs as $U_2, U_3$ and $6$ CNOTs as entangling layer, SEA of Schmidt coefficient layer
	as subSEA, {\ansatz} with three ALTs as $U_i (i=1,2,3)$ and $3$ CNOTs as entangling layer, {\ansatz} with three ALTs as $U_i (i=1,2,3)$ and $6$ CNOTs as entangling layer, respectively. $j$($j=3, 6$) CNOTs means that we set a composition of $j \CNOT$s controlled and targeted on the qubit-pairs $\{(i, N+i)\}_{i=0}^{j-1}$.}
\label{fig:VQE-LiH}
\end{figure}

\begin{figure}[htbp]
\centering
\includegraphics[width=0.45\textwidth]{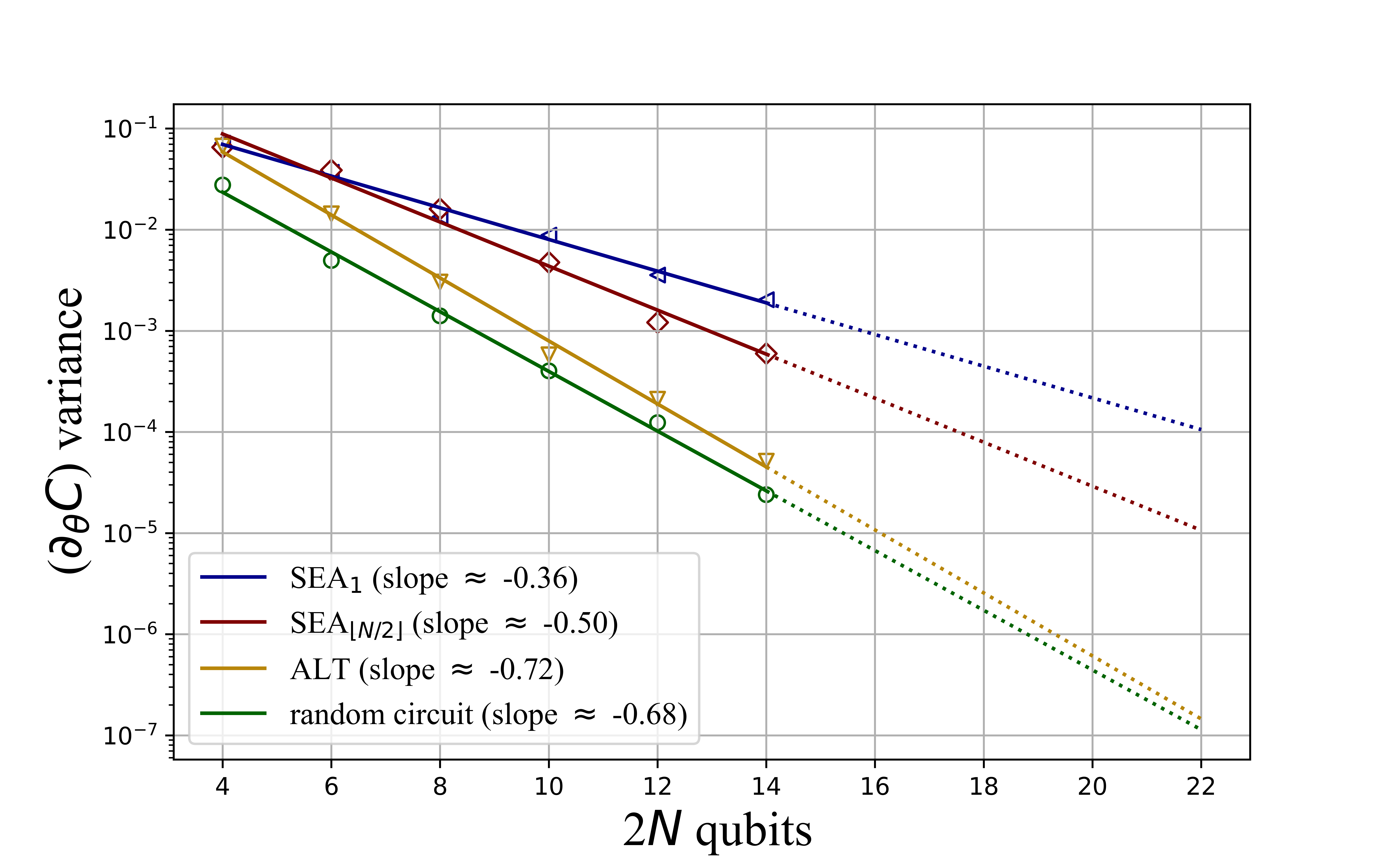}
\caption{\textbf{Comparison of the scaling of variance between different ansatzes on the Heisenberg model}. It shows the semi-log plot of the variance of the largest partial derivative among parameters in each round of sampling. We ensure different ansatzes have similar number of parameters by setting different depth. The solid part of the fitted lines represents the range we experimented with, while the dotted part represents the expected performance on a larger range.}
\label{Fig: BP-sub.2}
\end{figure}

\end{document}

%% file: pretex.tex
\usepackage[utf8]{inputenc}
\usepackage{mathrsfs}
\usepackage[T1]{fontenc}
\usepackage[qm]{qcircuit}

\usepackage{algorithm,algorithmic}

\usepackage[shortlabels]{enumitem}
\usepackage{epic,eepic,epsfig,latexsym,verbatim}
\usepackage{xcolor,color}
 
\usepackage{amsmath,amsfonts,amssymb}       

\usepackage{bbm}
\usepackage{bm}
\usepackage{times}
\usepackage{tcolorbox}
\usepackage{relsize}
\usepackage{graphicx}
\usepackage{multirow}
\usepackage{subfigure}

\usepackage{mathtools}
\usepackage[marginal]{footmisc}
\usepackage{url}

\usepackage[strict]{changepage}

\usepackage{hyperref}
\usepackage{cleveref}
\hypersetup{colorlinks=true,citecolor=blue,linkcolor=blue,filecolor=blue,urlcolor=blue,breaklinks=true}

\usepackage{theorem}
\newtheorem{definition}{Definition}
\newtheorem{proposition}{Proposition}

\newtheorem{lemma}[proposition]{Lemma}

\newenvironment{proof}{\noindent \textbf{{Proof~} }}{\hfill $\blacksquare$}

\newcounter{remark}

\newcounter{example}

\definecolor{darkblue}{RGB}{0,76,156}
\definecolor{darkkblue}{RGB}{0,0,153}
\definecolor{blue2}{RGB}{102,178,255}
\definecolor{darkred}{RGB}{195,0,0}

\newcommand{\nc}{\newcommand}
\nc{\rnc}{\renewcommand}

\newcommand{\bra}[1]{\langle#1|}
\newcommand{\ket}[1]{|#1\rangle}
\newcommand{\ketbra}[2]{|#1\rangle\!\langle#2|}
\newcommand{\braket}[2]{\langle#1|#2\rangle}
\newcommand{\braandket}[3]{\langle #1|#2|#3\rangle}

\newcommand{\CNOT}{\operatorname{CNOT}}

\nc{\proj}[1]{| #1\rangle\!\langle #1 |}
\nc{\avg}[1]{\langle#1\rangle}
\nc{\rank}{\operatorname{Rank}}
\nc{\smfrac}[2]{\mbox{$\frac{#1}{#2}$}}
\nc{\tr}{\operatorname{Tr}}
\nc{\ox}{\otimes}
\nc{\dg}{\dagger}
\nc{\dn}{\downarrow}
\nc{\csupp}{{\operatorname{csupp}}}
\nc{\qsupp}{{\operatorname{qsupp}}}
\nc{\E}{{\mathbb{E}}}
\nc{\var}{{\operatorname{Var}}}
\nc{\rar}{\rightarrow}
\nc{\lrar}{\longrightarrow}
\nc{\polylog}{{\operatorname{polylog}}}
\nc{\wt}{{\operatorname{wt}}}
\nc{\supp}{{\operatorname{supp}}}

\nc{\argmin}{{\operatorname{argmin}}}

\nc{\id}{{\operatorname{id}}}

\nc{\CHSH}{{\operatorname{CHSH}}}

\RequirePackage[framemethod=default]{mdframed}
\newmdenv[skipabove=7pt,
skipbelow=7pt,
backgroundcolor=darkblue!15,
innerleftmargin=5pt,
innerrightmargin=5pt,
innertopmargin=5pt,
leftmargin=0cm,
rightmargin=0cm,
innerbottommargin=5pt,
linewidth=1pt]{tBox}